\documentclass[aps,pra,showpacs,amssymb,nofootinbib,superscriptaddress,twocolumn, compress]{revtex4-1}
\usepackage{graphicx}

\usepackage[utf8]{inputenc}

\usepackage{amsmath}
\usepackage{amssymb}
\usepackage{amsthm}

\newtheorem{lemma}{Lemma}

\usepackage{csquotes}
\usepackage{ulem}
\usepackage{color}
\definecolor{s1}{rgb}{0.880722, 0.611041, 0.142051}
\definecolor{s2}{rgb}{0.368417, 0.506779, 0.709798}
\definecolor{n1}{rgb}{0.560181, 0.691569, 0.194885}
\definecolor{n2}{rgb}{0.922526, 0.385626, 0.209179}

\newcommand{\id}{\mathds{1}}

\usepackage{orcidlink}

\usepackage{hyperref}
\hypersetup{colorlinks=true,linkcolor=blue, linktocpage}

\usepackage{bm, physics,dsfont}
\newcommand{\cF}{\mathcal{F}}
\newcommand{\ii}{\mathrm{i}}
\begin{document}

\title{Optimal distributed multiparameter estimation in noisy environments}
\author{Arne Hamann\,\orcidlink{0000-0002-9016-3641}}
 \affiliation{Universit\"at Innsbruck, Institut f\"ur Theoretische Physik, Technikerstra{\ss}e 21a, 6020 Innsbruck, Austria}

\author{Pavel Sekatski\,\orcidlink{0000-0001-8455-020X}}
 \affiliation{Department of Applied Physics University of Geneva, 1211 Geneva, Switzerland}

\author{Wolfgang D\"ur\,\orcidlink{0000-0002-0234-7425}}
 \affiliation{Universit\"at Innsbruck, Institut f\"ur Theoretische Physik, Technikerstra{\ss}e 21a, 6020 Innsbruck, Austria}

\date{\today}

\begin{abstract}
We consider the task of multiple parameter estimation in the presence of strong correlated noise with a network of distributed sensors. We study how to find and improve noise-insensitive strategies.
We show that sequentially probing GHZ states is optimal up to a factor of at most 4. This allows us to connect the problem to single parameter estimation, and to use techniques such as protection against correlated noise in a decoherence-free subspace, or read-out by local measurements.
\end{abstract}

\maketitle
\section{Introduction}

Measurements lie at the heart of natural sciences, and play a fundamental role in all disciplines of physics and beyond. Quantum metrology offers the possibility to measure physical quantities of interest with quadratically enhanced precision as compared to classical approaches. This offers a huge potential for upcoming quantum technologies as an important tool to enhance sensing devices. When combining several sensors to a sensor network, quantities with different spatial or temporal dependence become directly accessible, offering increased freedom and possibilities not available in local sensing schemes. For a fixed quantity of interest with a specific spatial dependence, one can design a quantum metrology protocol in such a way that one estimates the quantity with maximal precision. This typically involves the preparation of a specific, multipartite entangled state distributed over multiple sensors located at different positions, its free evolution to imprint the quantity of interest, and the measurement of the state. The choice of state determines the quantity that can be sensed, while at the same time one can be insensitive to other signal or noise sources provided they have a different spatial dependence. This has been utilized in~\cite{sekatskiOptimalDistributedSensing2020, wolkNoisyDistributedSensing2020} to obtain protocols that are insensitive to multiple noise sources, or even noise originating from whole regions in space.

Here we consider the problem of sensing multiple such non-local quantities of interest (or linear functions thereof) simultaneously, i.e. a multi-parameter sensing problem. One may encode all quantities of interest in some initial state, which is then measured. However, in many relevant cases all parameters can not be optimally encoded or read out simultaneously. This makes it difficult to design optimal multi-parameter sensing protocols, i.e. find optimal initial states and measurement strategies that are generally applicable. In this work we restrict ourselves to commuting classical fields with different spatial dependencies, e.g. the different components of a Taylor or Fourier series, or distant-dependent fields emitted by sources located at various positions. In addition, we assume the presence of multiple strong noise sources that are spatially correlated, however with a different spatial dependence than the signals. This distinguishes our scenario from~\cite{rubioQuantumSensingNetworks2020,bringewattProtocolsEstimatingMultiple2021}, and requires the usage of sensing strategies that can cope with such imperfections.

Surprisingly, we find that a sequential approach where linear combinations of the parameters are measured individually is close to optimal. This is true for arbitrary figures of merit where quantities are weighted accordingly, and for particular choices one can indeed construct optimized protocols. Importantly, this allows one to use simple single-parameter sensing schemes that are based on the usage of GHZ states.
In particular, one can construct schemes that utilize a decoherence free subspace~\cite{sekatskiOptimalDistributedSensing2020} which allows for full elimination of spatially correlated noise, while at the same time maintaining the capability to sense quantities of interest. The single parameter sensing scheme has the additional advantage that the required initial state can be efficiently prepared and is the same in terms of required entanglement structure for all quantities of interest - up to some additional local operations before and during the sensing process. Furthermore, the readout can be implemented locally, i.e. no complicated entangled measurement is required, but local measurements of the individual qubits suffice. This shows that a simple, experimentally feasible scheme is almost optimal, and allows one to deal with relevant noise and imperfections efficiently. We also show that amongst all noise-insensitive strategies these sequential GHZ strategies are optimal up to at most a factor, which does not dependent on the number of sensors, signals, and noise sources.

The paper is organized as follows. We provide some background information of multiparameter sensing problems and the quantum fisher information matrix in Sec. \ref{sec:background}. There we also describe the problem setting we consider, review some known methods and provide new methods to compute quantities of interest. In Sec. \ref{sec:methods} we describe and analyze a sequential strategy, and discuss how it can be optimized w.r.t. different figures of merit. We prove the optimality of such sequential strategies over generic approaches within the decoherence free subspace in Sec. \ref{sec:results}, by showing that for any general protocol there always exists a sequential strategy that perform equal or better. We illustrate our results with help of some examples in Sec. \ref{sec:example} and summarize and conclude in Sec. \ref{sec:conclusion}.

\section{Background}\label{sec:background}

Quantum metrology investigates how to use quantum systems to estimate properties of physical systems. In recent years, quantum sensor networks gained increasing interest as a natural extension~\cite{sidhuGeometricPerspectiveQuantum2020} of multi-parameter quantum metrology~\cite{giovannettiAdvancesQuantumMetrology2011} and as a promising application for quantum networks.
They can be used for estimating local parameters~\cite{knottLocalGlobalStrategies2016,proctorMultiparameterEstimationNetworked2018}, function of local parameters~\cite{eldredgeOptimalSecureMeasurement2018,qianHeisenbergscalingMeasurementProtocol2019,geDistributedQuantumMetrology2018,rubioQuantumSensingNetworks2020,bringewattProtocolsEstimatingMultiple2021} and field properties~\cite{sekatskiOptimalDistributedSensing2020,hamannApproximateDecoherenceFree2022,wolkNoisyDistributedSensing2020,qianOptimalMeasurementField2021}. 
There are at least two cases where entanglement as a resource is helpful. First, it allows for estimating a global property (e.g. the average value, gradients and other analytical functions~\cite{qianHeisenbergscalingMeasurementProtocol2019}) with Heisenberg-scaling precision. However, Heisenberg-scaling collapses to a constant improvement in the presence of noise. 
Second, entanglement again helps to remove correlated noise in quantum sensor networks by constructing (approximate) decoherence-free subspaces and restores the Heisenberg-scaling~\cite{sekatskiOptimalDistributedSensing2020, wolkNoisyDistributedSensing2020, hamannApproximateDecoherenceFree2022}.

On the other hand, entanglement is not helpful, when all local parameters should be estimated~\cite{knottLocalGlobalStrategies2016}. Consequently, it is not useful for estimating many global properties~\cite{rubioQuantumSensingNetworks2020, bringewattProtocolsEstimatingMultiple2021}, which allows for recovering the local parameters.

Here we consider the scenario, where a distributed network of (entangled) sensors is used to estimate particular features of space-dependent fields.
We generalize the decoherence free subspaces to the multi-parameter distributed sensing setting. We focus on the asymptotic ``Fisher" regime of many repetitions, where the performance of a sensing protocol is well characterized by the Fisher information matrix (FIM), that we now introduce.

\subsection{The Fisher Information Matrix}

Let a single run of the experiment be described by a random variable $X$, distributed accordingly to the parametric model $P(x|\bm \alpha)$ where $\bm \alpha =(\alpha_1,\dots \alpha_s)$  is the column vector collecting the parameters of interest. The FIM matrix $F(\bm \alpha)$ associated by to such protocol is given by 
\begin{equation}\begin{split}
    F_{\mu \nu}(\bm \alpha) &= \mathds{E}\left[\partial_\mu \ln p(x|\bm \alpha) \partial_\nu \ln p(x|\bm \alpha)\right] \\
    &=\sum_x \frac{\partial_\mu p(x|\bm \alpha) \partial_\nu p(x|\bm \alpha)}{ p(x|\bm \alpha)},
\end{split}
\end{equation}
with $\partial_\mu p(x|\bm \alpha) \equiv \frac{\partial p(x|\bm \alpha)}{\partial \alpha_\mu}$. 
The FIM is positive semi-definite and defines a Riemannian ``statistical" metric over the set of parameters. It has found various applications. In parameter estimation its significance is well emphasized by the Cramer-Rao bound
\begin{equation} \label{eq: CRB}
    \text{Cov}(\hat {\bm \alpha}) \geq F^{-1}(\bm \alpha),
\end{equation}
relating the FIM to the covariance matrix $\text{Cov}(\hat {\bm \alpha})$ of an estimator $\hat {\bm \alpha}(X)$.
Recall that for unbiased estimators, i.e. such that $\mathds{E}[\hat {\bm \alpha}(X)] = \bm \alpha$, the covariance matrix is given by 
\begin{equation}\label{def: covariance}
\text{Cov}_{\mu \nu}(\hat{\bm \alpha})= \mathds{E}\Big[(\hat \alpha_\mu(X) - \alpha_\mu)(\hat \alpha_\nu(X) - \alpha_\nu)\Big],
\end{equation}
and is a natural quantifier of estimation errors. Furthermore, with a proper choice of estimators the  Cramer-Rao bound is asymptotically saturable ~\cite{ragyCompatibilityMultiparameterQuantum2016, grossOneManyEstimating2021}. 

For a reader familiar with scalar parameter estimation theory it may be insightfull to consider small parameter variations $\dd \bm \alpha = ( v_1 \dd \theta, \dots, v_s \dd \theta)^\intercal$ along a fixed direction $\bm v = (v_1,\dots v_s)^\intercal$. Here, the quantity $F_{(\bm v)}(\theta) \equiv \bm v^\intercal F(\bm \alpha) \bm v$ becomes the scalar Fisher information for the parameter $\theta$, along a curve $\bm \alpha(\theta)$ passing through the point $\bm \alpha$. In particular, it can be understood as the susceptibility of the Kullback–Leibler divergence $F_{(\bm v)} = \lim_{\dd \theta \to 0}\frac{1}{\ln 4}\frac{D\big(p(x|\bm \alpha)||p(x|\bm \alpha +\dd \bm \alpha)\big)}{\dd \theta^2}$~\cite{cover1991information}. It is important to emphasise that the quantity $F_{\bm v}$ can not be used to tightly bound the error of estimating $\alpha_{(\bm v)} \equiv \bm v^\intercal \bm \alpha$ in the \textit{multi-parameter setting}, because one does not have the knowledge of the curve $\bm \alpha(\theta)$ to start with. See~\cite{grossOneManyEstimating2021} for a detailed discussion of the relation between scalar and multi-parameter estimation settings.

In a quantum experiment the parameters are first encoded into a quantum state $\rho_{\bm \alpha}$, and the classical setting is only recovered upon fixing the final measurement performed on the system. Nevertheless, it is possible to define the quantum Fisher information matrix (QFIM) $\mathcal{F}(\rho_{\bm \alpha})$ that  only depends on the parametric state model
\begin{equation}\begin{split}    
    \mathcal{F}_{\mu \nu}(\rho_{\bm \alpha}) = \frac{1}{2} \tr \rho_{\bm \alpha} (L_\mu L_\nu+L_\nu L_\mu),
\end{split}
\end{equation}
where the symmetric logarithmic derivative (SLD) operators $L_\mu$ and $L_\nu$ are solutions of 
\begin{equation}
    \partial_\mu \rho_{\bm \alpha} = \frac{1}{2}\left(L_\mu \rho_{\bm \alpha} +  \rho_{\bm \alpha}  L_\mu \right)
\end{equation}
 with $\partial_\mu \rho_{\bm \alpha} \equiv \frac{\partial \rho_{\bm \alpha}}{\partial \alpha_\mu}$. The QFIM is also positive semi-definite and defines a metric in the space of quantum states parameterized by $\rho_{\bm \alpha}$~\cite{LiuYuanLuWang2019}. Moreover, it satisfies
\begin{equation}\label{eq: QFIM>FIM}
\mathcal{F}(\rho_{\bm \alpha}) \geq F(\bm \alpha),    
\end{equation}
where $F(\bm \alpha)$ is the FIM arising from any measurement choice. In contrast to the scalar case, here a measurement that saturates the above equation does not always exist. However, for any fixed direction $\dd \bm \alpha = ( v_1 \dd \theta, \dots, v_s \dd \theta)$ the bound 
\begin{equation}
\mathcal{F}_{(\bm v)}(\rho_{\bm \alpha}) \equiv \bm v^\intercal \mathcal{F}(\rho_{\bm \alpha})  \bm v \geq F_{(\bm v)} (\bm \alpha)
\end{equation}
can be attained by the right measurement choice. The impossibility of saturating Eq.~\eqref{eq: QFIM>FIM} arises when the measurements maximizing $F_{(\bm v)} (\bm \alpha)$ for different directions are not compatible. For completeness, in the appendix~\ref{app: QFIM} we present a derivation of the equation~\eqref{eq: QFIM>FIM} and briefly discuss its saturability. Combining the inequalities in the Eq.\eqref{eq: CRB} and Eq.~\eqref{eq: QFIM>FIM} gives the so called quantum Cramer-Rao bound
\begin{equation}\label{eq: QCRB}
     \text{Cov}(\hat {\bm \alpha}) \geq \cF^{-1}(\rho_{\bm \alpha}).
\end{equation}

\subsection{Comparing sensing protocols}

In the following we will usually deal with pure states $\rho_{\bm \alpha} =\ketbra{\psi_{\bm \alpha}}$, and with a fixed unitary parameter encoding $\ket{\psi_{\bm \alpha}} = e^{-\ii\, \bm \alpha^\intercal \bm G} \ket{\psi}$, where $\bm G$ contains the generators of the parameters. Hence, we will simply denote the QFIM as $\mathcal{F}_\psi$, where $\ket{\psi}$ is the intial state in which the probe (sensing systems) are prepared.

The performance of a multi-parameter estimation strategy is  quantified by the covariance matrix $\text{Cov}(\hat{\bm \alpha})$ in the Eq.~\eqref{def: covariance}. The quantum Cramer-Rao bound~\eqref{eq: QCRB} lower-bounds the covariance matrix with the QFIM, which can be readily computed from the initial state $\ket{\psi}$. Hence, it is natural and convenient to use the QFIM as a partial ordering of the estimation strategies. Accordingly to this ordering, a strategy with probe state $\left\vert \psi \right\rangle$ is not worse than $\left\vert \phi \right\rangle$, iff the corresponding QFIM fulfills $\mathcal{F}_\psi \geq \mathcal{F}_\phi \Leftrightarrow \mathcal{F}_\psi - \mathcal{F}_\phi \geq 0$.  
This is a partial order i.e. not all strategies are comparable and hence there is generally no global optimal strategy. However, a strategy is extremal if there is no strategy with greater QFIM. Note that generally the ordering of QFIM does not imply an ordering of the corresponding FIM. Nevertheless in our context we will later show that the strategies identified as extremal have saturable QFIM (by constructing a measurement such that $F(\bm \alpha)= \cF_\psi$), guaranteeing that they are also improved with respect to a partial ordering induced by the FIM.

A different approach to compare strategies is by a figure of merit. The (weighted) trace 
\begin{equation}\label{equ:merit}
    M=\tr W  \mathrm{Cov}( \hat{\bm \alpha}) =\sum_i W_{ij} \mathrm{Cov}_{ij}(\hat{\bm \alpha})
\end{equation} of the covariance matrix with respect to a positive semidefinite matrix $W$ is a broadly used figure of merit for multi-parameter estimation. In particular, one ofter often consider the case where $W $ is diagonal leading to $ M=\sum_i w_i \mathrm{Var}( \alpha_i)$. A figure of merit induces a total order and makes all strategies comparable. Hence now there exits an optimal strategy for this figure of merit. Multiparamter estimation without noise under this figure of merit is well studied~\cite{proctorMultiparameterEstimationNetworked2018, bringewattProtocolsEstimatingMultiple2021, ehrenbergMinimumEntanglementProtocols2021, rubioQuantumSensingNetworks2020}. Notice that finding an optimal strategy for a given figure of merit and finding extremal strategies with respect to FIM are related but independent problems.

\subsection{Setup}
\begin{figure}
    \includegraphics[width=\linewidth]{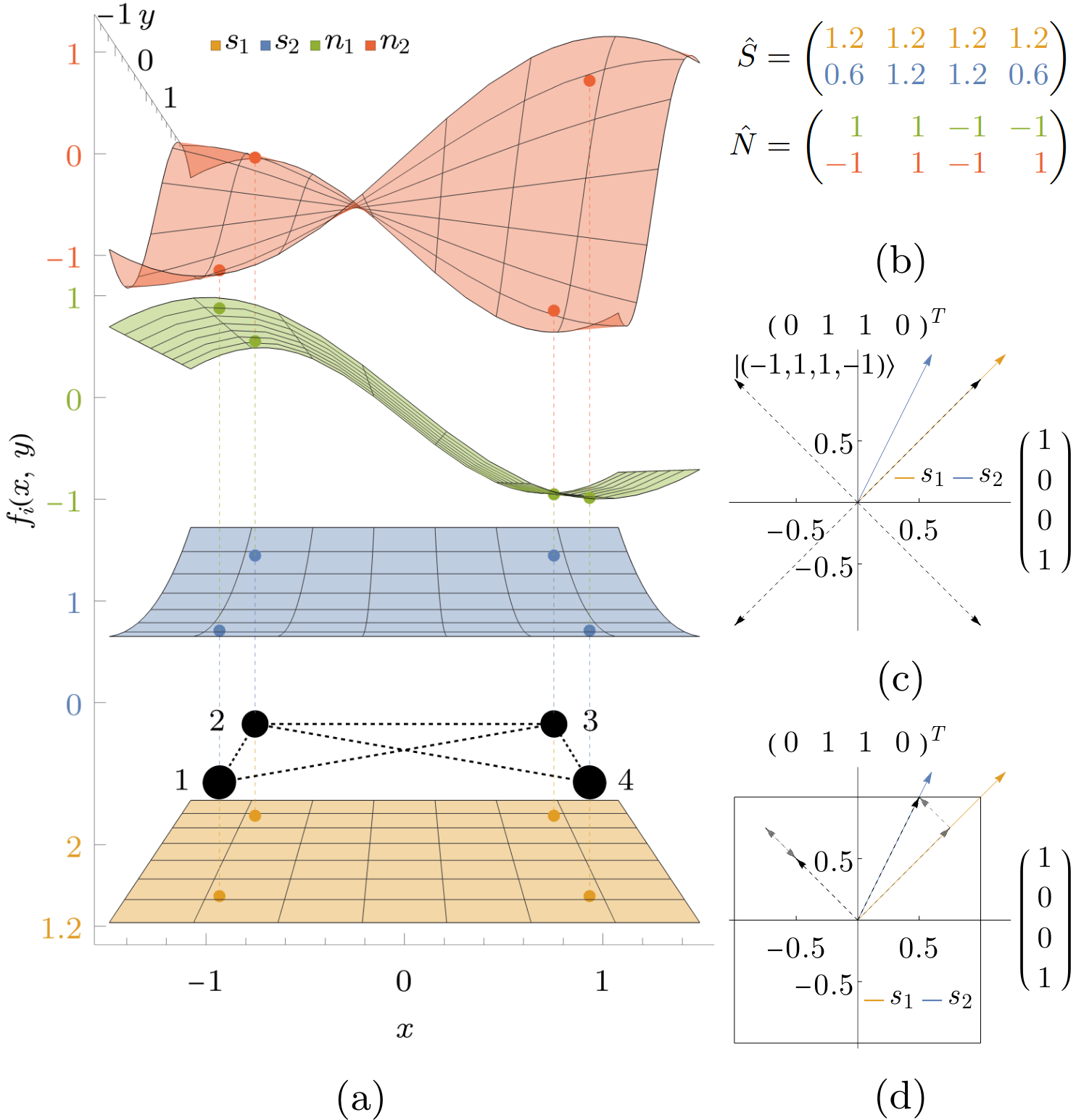}
    \caption{ (a) A quantum sensor network consisting of for sensors (black spheres) interacting with two signal $s_1=1.2,s_2=0.6 (1+(\frac{y + 1}{2})^2)$ and two noise fields $n_1=-\sin(\frac{\pi}{2}x),n_2=-\sin(xy)/\sin(1)$. 
    (b) In the signal matrix $\hat{S}$ and noise matrix $\hat{N}$ each row corresponds to one field and each column to one sensor.
    (c) A cut along the (0,1,1,0) and (1,0,0,1) axis through the 4-dimensional sample space. Coloured vectors represent the signals. Black Dashed vectors represent 4 computational basis states $\left\vert -1,1,1,-1\right\rangle, \left\vert 1,-1,-1,1\right\rangle, \left\vert 1,1,1,1,1\right\rangle $ and $\left\vert -1,-1,-1,-1\right\rangle$. The energy of the state with respect to a certain field is given by the scalar product of their vector representations in the sample space. As the shown subspace is orthogonal to the noises, all represented states have the energy (zero) with respect to the noise fields. Therefore the cut is the representation of a decoherence-free subspace.
    (d) State vectors with non-integer entries, do not have a corresponding quantum state. Anyhow these states can effectively be implemented, by first decomposing them into computational basis states, and then applying $X$ pules at the appropriate sensor and time. In the example pulses are applied after three-quarters of the interaction time to the first and the last sensor to align with the blue signal, or at all qubits to reduce the effective energy by $1/2$. Therefore all state vectors within the black square can be efficiently simulated.  }
    \label{fig:overview}
\end{figure}

\begin{figure}
    \includegraphics[width=\linewidth]{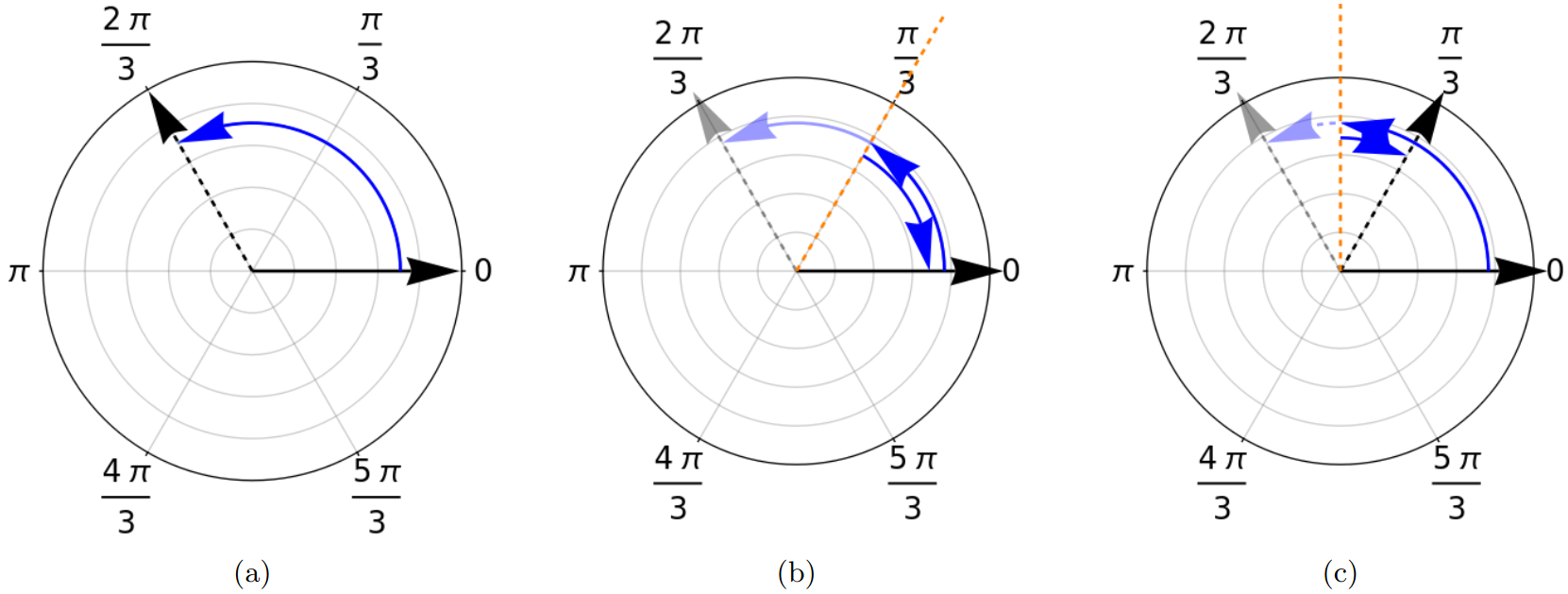}
    \caption{Flipping represented in the $X-Y$ plane of the Bloch sphere. The Qubit is initialized in zero. (a) Without flipping the the qubit will accumulate a certain phase. (b) After the half evolution a flip ($X$-gate) is applied. This effectively inverts the direction of the evolution, such that after the full evolution effectively no phase was accumulated. (c) A flip after three quaters of the evolution will effectively give half the phase.}
    \label{fig:flipping}
\end{figure}

We consider commuting fields with different spatial dependence  $f_i(\vec x)$, where the parameters of interest are not local quantities, but field amplitudes $\alpha_i$.  For example, the mass of an object is a multiplicative factor to the gravitational field it emits, idem for the magnetization of a magnet.
A quantum sensor network, i.e. quantum sensors distributed in space connected by a quantum network, is used to estimate these parameters (Fig.~\ref{fig:overview}). We further assume that individual sensors are described by qubits and that the interaction between the field and the sensor is locally generated by the $\sigma_z$ operator.
The parameters are then encoded by the generators
$$G_i= \sum_{j=1}^n f_i(\vec{x}_j)\sigma_z^{j} = \sum_{j=1}^n \hat{S}_{ij}\sigma_z^{j},$$
via a unitary interaction governed by he Hamiltonian $H_\text{signal} = \sum_i \alpha_i G_i$. These fields are called signals, for fixed sensor positions all relevant spatial dependencies are contained in the matrix $\hat{S}$ with components $S_{ij}= f_i(\vec{x}_j)$ (Fig.:~\ref{fig:overview}b).

Correlated noise is modeled by fields $f^{\text{noise}}_i$ with global amplitudes (strength) $\boldsymbol{\beta}=(\beta_1,\beta_2,...)$. The spatial noise dependencies are contained in the noise matrix $\hat{N}_{ij}=f^{\text{noise}}_i(\vec{x}_j)$ and the generators $G_i^\text{noise}$ are defined analogously to the signals. It is possible to consider several models of how the noise enters in the estimation scenario. These range from fixed but unknown strength $\bm \beta$, that one does not want to estimate given that it would cost additional resources. To the worst case scenario, where time-dependent amplitudes $\bm \beta$ are subject to white noise fluctuations. In any case our strategy allows one to completely obliterate the effect of the noise by preparing and keeping the state of the sensor network inside a noise-insensitive subspace. This setup was introduced in~\cite{sekatskiOptimalDistributedSensing2020} and well studied for single parameter estimation~\cite{sekatskiOptimalDistributedSensing2020, wolkNoisyDistributedSensing2020, hamannApproximateDecoherenceFree2022}.

\subsection{Decoherence free subspaces}
In this paper we use the joined eigenbasis  $\sigma_z^j\left\vert \boldsymbol{k} \right\rangle=k_j\left\vert \boldsymbol{k} \right\rangle$ of local $\sigma_z$ as basis.
This basis has the nice property that the energy with respect to generator $G_i$ is given by 
\begin{equation}\label{equ:energy}
    G_i^\mathrm{eff}\left\vert\bm k\right\rangle  =\sum_i \hat S_{ij} k_j  \vert \bm k \rangle = \left(\boldsymbol{e}_i^\intercal\hat{S} \boldsymbol{k} \right)\vert \bm k \rangle
\end{equation}
the scalar product of the $i$-th row of the matrix $\hat{S}$ and the vector $\bm k$ labeling the state.
The optimal state to measure a single parameter or a linear combinations $\alpha_{(\bm v)}= \bm v ^\intercal \bm \alpha$ is a GHZ state given by
\begin{equation}\label{eq: GHZ def}
    \left\vert \mathrm{GHZ}_{\boldsymbol{k}}\right\rangle = \frac{\left\vert\boldsymbol{k}\right\rangle+\left\vert-\boldsymbol{k}\right\rangle}{\sqrt{2}}.
\end{equation}
with single parameter QFI of $4T^2 (\bm v ^\intercal \hat{S} \bm k)^2$.
For GHZ state the optimal measurement saturating the QFI is local, cf. \ref{app: GHZ local meas}.
Additionally, as illustrated in Fig.~\ref{fig:flipping} non integer entries of $k_j$ can be effectively implemented~\cite{sekatskiOptimalDistributedSensing2020}. This is done by inverting parts of the evolution  by flipping the local qubits at time $t_j=\frac{1+k_j}{2}$. Notice that only a single flip per qubit is required. We also remark that other methods are available to obtain a reduced coupling of individual spins to the field, e.g. by using different energy levels.

The state is protected from noise if $\vert \boldsymbol{k}\rangle$ and $\left\vert -\boldsymbol{k}\right\rangle$ have the same eigenvalues for all the noise generators $G_i^{\mathrm{noise}}$. In particular, this is the case if the vector $\boldsymbol{k}$ is the in the kernel of the noise matrix $\hat{N}\boldsymbol{k} = 0$, leading to the definition of the decoherence-free subspace introduced in~\cite{sekatskiOptimalDistributedSensing2020}
\begin{equation}\label{eq: DFS}
    \mathrm{DFS} = \left\{\boldsymbol{k} \big| \hat{N}\boldsymbol{k} = 0,\, \Vert \boldsymbol{k}\Vert_\infty\leq 1\right\}.
\end{equation}
The DFS is a real $(n-\mathrm{rank} \hat{N})$ dimensional polytope (Fig.~\ref{fig:overview}c,d and Fig.~\ref{fig:example-non-integer}). Notice that by construction if $\boldsymbol{k}$ is in the DFS then $-\boldsymbol{k}$ is too.
Thus, all GHZ states in Eq.~\eqref{eq: GHZ def} with $\boldsymbol{k}\in \mathrm{DFS}$ are protected from the noise and remain pure during the evolution.

More generally, one may also consider the affine decoherence free subspaces $\text{DFS}_{\bm \kappa} =\{\bm k | \hat{N}\bm k = \bm \kappa,\, \|\bm k\|_\infty\leq 1\}$. All states within an $\text{DFS}_{\bm \kappa}$ are also protected from noise, however, in contrast to the GHZ states, they are in general very hard to engineer and require entangled measurements for readout. For these reasons in the rest of the paper we only consider noise-insensitive strategies involving states from the DFS defined in Eq.~\eqref{eq: DFS}. Nevertheless, in the appendix~\ref{appendix:affineDFS} we present an example for which we prove that it can be advantageous to prepare states in a well chosen affine DFS; this is is sharp contrast with the single-parameter regime where a particular GHZ state with $\hat{N}\bm k=0$ is always optimal~\cite{sekatskiOptimalDistributedSensing2020}. In the appendix~\ref{appendix:affineDFS:upperbound} we upper-bound the advantage of the affine DFS to a factor four at most.

\section{Methods}\label{sec:methods}
\subsection{Sequential GHZ strategy}
A sequential GHZ strategy~\cite{altenburgMultiparameterEstimationGlobal2019} consists of preparing different states  $\ket{\mathrm{GHZ}_{\bm k_i}}$ in different rounds of the experiment. Let us denote the frequencies with which different states are prepared by $r_i$, which sum up to 1. This strategy has the same QFIM as the one where one repeatedly prepare the state  
\begin{equation}
   \varrho = \bigoplus_i r_i \ketbra{\mathrm{GHZ}_{\bm k_i}}
\end{equation}
where the block-diagonal structure is carried by a classical label denoting which GHZ state has been prepared. In the limit of a large number of repetitions the two strategies are equivalent, nevertheless it is sometimes more convenient to argue in terms of a single state $\varrho$.

The sequential strategy is not necessarily optimal. In particular, it is generally better to measure linear combinations of the parameters~\cite{bringewattProtocolsEstimatingMultiple2021}.

\subsubsection{QFIM}
At the core of our analysis is the QFIM.
For sequential GHZ strategies within the DFS it takes a particularly simple form, and can be directly expressed in terms of the vectors $\bm k_i$ and the signal matrix $\hat S$ as
\begin{equation}\label{equ:K1}
    \mathcal{F} = 4 T^2 \, \Hat{S} \hat{K}\hat{S}^\intercal
\text{, where }
    \hat{K} = \sum_i r_i \boldsymbol{k}_i \boldsymbol{k}_i ^\intercal
\end{equation}
is the weighted sum over the projectors onto $\boldsymbol{k}_i$.
This can be seen by looking at the expectation value of the SLD. For a sequential strategy the state and the SLD are block-diagonal and hence the QFIM is given by the average $\cF = \sum_i r_i \mathcal{F}_i$ over all states $\ket{\mathrm{GHZ}_{\bm k_i}}$. In turn, by virtue of Eq.~\eqref{equ:energy} each of these states is insensitive to signals orthogonal to $\bm k$, and has the QFIM given by $\mathcal{F}_{\vert \mathrm{GHZ}_{\bm k}\rangle}= 4 T^2 \hat{S}\bm k \bm k^\intercal\hat{S}^\intercal$ as shown in appendix~\ref{appendix:QFIM}, leading to  Eq.~\eqref{equ:K}.

It is worth emphasizing that the QFIM in Eq.~\label{equ:K} is saturable with local measurements. This is because in each run one only estimates a single parameter, and for the states $\vert \mathrm{GHZ}_{\bm k}\rangle$ it is known~\cite{sekatskiOptimalDistributedSensing2020} that the QFI can be saturated with a measurement strategy which simply combines the results of fixed local measurements, see~\ref{appendix:QFIM}.

\subsection{General strategies}
To show the optimality of the sequential GHZ strategies we compare them to general strategies, where an arbitrary superposition of states within the DFS is prepared 
\begin{equation}\label{eq: general psi}
    \vert \psi \rangle = \sum_{\bm k \in \kappa \subset \mathrm{DFS}} c_{\bm k }\ket{\bm k},
\end{equation}
where all the states $\ket{\bm k}$ are orthogonal.

 This strategies might not be directly implementable as $\kappa$ might contain more than $2^n$ elements. Therefore $\bm k,\bm k'\in \kappa$ cannot all be prepared orthogonal. Additionally $\bm k$ and $\bm k'$ might require incompatible flipping sequences.
Anyhow it turns out that all these strategies can actually be effectively implemented with an auxiliary system. The states get distinguishable by preparing for each $\bm k$ orthogonal labels in the auxiliary system. These labels allow then apply incompatible flipping sequences, by controlling them on the auxiliary system being in the corresponding label state. 
This strategies include all possible initial states. The QFIM of a strategy represented by the state in Eq.~\eqref{eq: general psi} is given by
\begin{equation}\label{equ:K2} \begin{split}
    \mathcal{F}_\psi &= 4 T^2 \Hat{S} \hat{K}_\psi \hat{S}^\intercal
\text{, with } \\
    \hat{K}_\psi &=\! \sum_{\bm k \in \kappa }\vert c_{\bm k}\vert ^2  \boldsymbol{k}\boldsymbol{k} ^\intercal - \bar{\bm k}\bar{\bm k}^\intercal  =\! \sum_{\bm k \in \kappa} \vert c_{\bm k}\vert ^2 (\bm k -\bar {\bm k})  (\bm k -\bar {\bm k})^\intercal
\end{split}
\end{equation}
and $\bar{\bm k} = \sum_{\bm k \in \kappa } \vert c_{\bm k}\vert^2 \bm k$. 

\section{Results}\label{sec:results}

\subsection{The pre-order of strategies is independent of the signals $\hat S$}

Our goal now is to use the partial order provided by the QFIM to compare\footnote{In the case of strategies this only induces a pre-order, as different strategies may lead to the same QFIM.} general noise-insensitive strategies labeled by a state $\varrho$. We say that a strategy $\rho$ improves over another one $\sigma$ if their QFIM satisfy $\cF_\rho \geq \cF_\sigma$ (in slight abuse of language this includes the case where the QFIMs are equal). We call a strategy $\rho$ extremal if there exist no strategy $\sigma$ with $\cF_\sigma > \cF_\rho$.

The first general observation is that the sensing strategies can be directly compared on the level of  the $\hat K$ matrices in Eqs.~(\ref{equ:K1},\ref{equ:K2}). Indeed for any two strategies with $\hat K_1 \geq \hat K_2$, we get  $\hat{S} (\hat{K}_1-\hat {K}_2)\hat{S}^\intercal = \left(\hat{S} \sqrt{ \hat{K}_1-\hat {K}_2}\right)\left(\hat S \sqrt{\hat{K}_1-\hat {K}_2}\right)^\intercal \geq 0$ and thus  $\cF_1\geq \cF_2$ for all signal matrices $\hat S$.  Similar arguments can be made about extremal strategies, however for degenerate signal matrices $\hat S$ some strategies leading to extremal $\hat K$ might not be extremal for the QFIM.

\subsection{Extremal strategies are GHZ  sequential}

\label{sec: construction}

Consider any strategy represented by the  state 
\begin{equation}
\vert \psi \rangle = \sum_{\bm k \in \kappa \subset \mathrm{DFS}} c'_{\bm k}\vert \boldsymbol{k} \rangle
\end{equation}
leading the the QFIM $\cF_\psi$ of Eq.~\eqref{equ:K}.We now show that there is a sequential GHZ strategy $\varrho$ which outperforms $\ket{\psi}$, i.e. which leads to a larger QFIM $\cF_\varrho \geq \cF_\psi$. To do so we construct $\varrho$ following three steps that are easy to follow.
\begin{itemize}
    \item[1.]\textbf{Symmetrization:}
    Define a new state  $\vert \phi \rangle = \sum_{\bm k \in \kappa } \frac{c_{\bm k }}{\sqrt{2}} \ket{\bm k}$ with symmetrized coefficients $c_{\boldsymbol{k}}  =\sqrt{|c'_{\boldsymbol{k}}|^2+|c'_{\boldsymbol{-k}}|^2}$, which can be expressed as a superposition of GHZ states
    \begin{equation}
    \ket{\phi} = \sum_{\bm k \in \kappa \subset \mathrm{DFS}} c_{\bm k} \ket{\textrm{GHZ}_{\bm k}}.
        \end{equation}
We have $\hat{K}_\phi-\hat{K}_\psi = - \bar{\bm k}\bar{\bm k}^\intercal \leq 0$, which implies that the new state has an improved QFIM $\cF_\phi \geq \cF_\psi$.
   
    \item[2.]\textbf{Sequentiallization:}
    Instead of preparing the state $\ket{\phi}$,
    one can run a sequential strategy given by the state
    \begin{equation}\label{eq: sequentalization}
    \varrho = \bigoplus_{\bm k \in \kappa \subset \mathrm{DFS}} |c_{\bm k}|^2 \ketbra{\textrm{GHZ}_{\bm k}},
    \end{equation}
    and leading the the same QFIM $\cF_\phi = \cF_\varrho$.
\end{itemize}
 A more detailed discussion for the symmetrization and the sequentiallization steps can be found in  appendices~\ref{appendix:symmetrization} and \ref{appendix:sequential} respectively. At this point we have already constructed the sequential strategy that improves over the original one. Nevertheless, the following step allows to further improve the sequential strategy.

\begin{itemize}
    \item[3.]\textbf{Vertices:}
    For the sequential strategy $\varrho$ in Eq.~\eqref{eq: sequentalization} one can decompose each vector $\boldsymbol{k}=\sum_i p^{\boldsymbol{k}}_i \boldsymbol{v}^{\boldsymbol{k}}_i$ into a convex combination of vectors $\bm v_i$ pointing to the vertices of the DFS. This allows one to define a new sequential strategy 
    \begin{equation}\label{eq: seq vortex}
    \rho =\bigoplus_{\bm k,i}  r_{\bm k,i} \ketbra{\textrm{GHZ}_{\bm v_i^{\bm k}}}. 
    \end{equation}
    with $r_{\bm k,i} \equiv |c_{\bm k}|^2 p_i^{\bm k}$, where only GHZ states corresponding to the vertices of the DFS are prepared sequentially. In the appendix~\ref{appendix:vertices} we show that $\sum_i p_i^{\bm k}\, \bm v_i^{\bm k} \bm v_i^{\bm k \intercal}\geq \boldsymbol{k}\boldsymbol{k} ^\intercal$, which implies that step 3 is also beneficial for the QFIM $\cF_\rho \geq \cF_\varrho$.
    \end{itemize}

Finally, in the appendix~\ref{appendix:vertex-GHZ-extremal} we show that any such \textit{vertex-sequential} GHZ strategy $\rho$, i.e. a strategy which only involves GHZ states pointing to the vertices of the DFS, is extremal  -- there is no noise insensitive strategy $\varrho'$ with $\cF_{\varrho'}>\cF_\rho$. This is proven by assuming that such an improved strategy $\varrho'$ exists, and then showing that this leads to a contradiction with the assumption that all the GHZ states in $\rho$ point into vertices of the DFS. In summary we how established the following result.\\

\textbf{Result.} \textit{In the considered setting with commuting signal $G_i$ and noise $G_i^{\mathrm{noise}}$ generators, all extremal (with respect to the the QFIM) noise-insensitive strategies in the DFS of Eq.~\eqref{eq: DFS}, can be realised as vertex-sequential GHZ strategies $\rho$ in Eq.~\eqref{eq: seq vortex}. }\\

This allows on the one hand to use fixed states for all sensing problems, and on the other hand to profit from the advantages of the single-parameter strategy, namely to require only local measurements for readout.
In particular, the results shows that in the noiseless scenario where the signals are acquired in successive rounds with evolution time $T$, the sequential GHZ strategy is also optimal.

\subsection{Optimal sequential GHZ strategy}

In the previous section we have seen, that from the QFIM perspective is sufficient to consider sequential GHZ states pointing into the vertices of the DFS $\boldsymbol{v}_i$. Now let us focus on a particular figure of merit $M$ introduced in Eq.~\eqref{equ:merit}. The Cramer-Rao bound implies the bound $M \geq \tr W F^{-1}(\bm \alpha)$, which is saturable in the assymptotic sampling limit. Furthermore, for any strategy $\ket{\psi}$ we know that $F_\psi^{-1}(\bm \alpha)\geq \cF^{-1}_{\psi}\geq \cF^{-1}_{\rho}$ where $\rho$ is some vertex-sequential GHZ strategy. But we have seen that for such strategies the QFIM is saturable with a local measurement, i.e $F_\rho(\alpha) = \cF_{\rho}$. Hence we find that the bound
\begin{equation}
    M \geq \min_\rho \tr W \cF_\rho^{-1}
\end{equation}
is attainable in the limit of many repetitions. Here the minimization is taken over all vertex-sequential GHZ strategies, the optimal value of the figure of merit is thus given by
\begin{equation}\label{equ:opt}
    M_{\mathrm{opt}} = \min_{r_1,r_2,...}  M\left(\bigoplus_{i} r_i\ketbra{\mathrm{GHZ}_{\boldsymbol{v}_i}}\right)
\end{equation}
where one minimizes over the rates of sequentially probing all vertices of the DFS. In general, this expression is difficult to simplify further, because it involves the inverse of the average QFIM $\cF_\rho = \sum r_i\,  \mathcal{F}_{\vert \mathrm{GHZ}_{{\bm v}_i}\rangle}$.

If the $\bm v_i$ are orthogonal and $\hat{S}$ is invertible one obtains an expression of $\cF_\rho^{-1}$ that is linear in $r_i^{-1}$. Then, as shown in Appendix~\ref{appendix:orthogonal_labels}, the optimal rates are found to be
\begin{equation}
    r_i = \frac{|\bm v_i|^2}{\sqrt{w_i}}\left(\sum_j\frac{|\bm v_j|^2}{\sqrt{w_j}}\right)^{-1}
\end{equation}
 with  $w_i =   \left(\hat S^{-1} \bm v_i\right)^\intercal W \hat S^{-1}\bm v_i$.

\section{Example}\label{sec:example}
To illustrate our results, we will consider the example with a four sensor network shown in Fig.~\ref{fig:overview}. The sensor network is planar and forms a square. Denoting the spatial coordinates with $\vec x = (x,y)$, the sensors are located at positions $\vec x_i \in\{(-1,-1), (-1,1), (1,1), (-1,1)\}$. The two signals are chosen to have field shapes that are common in physics -- $\alpha_1$ is the strength of a constant field $f_1(\vec x)=1.2$, while $\alpha_2$ is the strength of a field constant along the x-axis and quadratic along y-axis $f_2(\vec x)=0.6(1+\frac{1}{4}(y+1)^2)$. They gives the following signal  matrix 
\begin{equation}
\hat{S}=\left(\begin{matrix}
    1.2& 1.2& 1.2& 1.2\\ 0.6& 1.2 & 0.6 & 1.2
\end{matrix}\right).
\end{equation} 
The signals have to be estimated in the presence of two periodic noise fields. The first one is periodic along the x and constant along the y axis $f^{\text{noise}}_1 =-\sin(\frac{\pi}{2}x)$, the second is a standing wave along both axis $f^{\text{noise}}_2 =-\sin(x y)/\sin(1)$. The noise matrix is then given by  \begin{equation}
\hat{N}=\left(\begin{matrix}
    1& 1& -1& -1\\ -1& 1 & -1 & 1
\end{matrix}\right).
\end{equation}
The decoherence free subspace is the convex polytope with vertices 
\begin{equation}\begin{split}
\bm v_1 = -\bm v_3 &= \left(1,1,1,1\right) \\
\bm v_2 = -\bm v_4 &=\left(-1,1,1,-1\right).
\end{split}
\end{equation}
The example is artificial in the sense that the considered signals and noise fields are not motivated from a particular set-up or situation, but serve to illustrate our approach. Notice that any signal and noises with different spatial dependencies can be treated in the same fashion.
We use this example to illustrate the optimization of a given strategy. Then we will compute the optimal strategy for the figure of merit given by the trace of the covariance matrix $M= \tr(\mathrm{Cov} (\hat {\bm \alpha})) =\mathrm{Var}(\hat \alpha_1)+ \mathrm{Var}(\hat \alpha_2)$.

\subsection{Optimizing a given strategy}
To illustrate the steps of the optimization procedure we consider the following initial state
\begin{align*}
    \left\vert\psi\right\rangle 
    = &\left(\nu\left\vert 1,1\right\rangle_{2,3}+\mu\left\vert -1,-1\right\rangle_{2,3} \right) \\
    &\otimes\frac{1}{\sqrt{2}}\left(  \left\vert \gamma,\gamma\right\rangle_{1,4}+ \left\vert -\gamma,-\gamma\right\rangle_{1,4}\right),
\end{align*}
with $\vert\mu\vert^2 + \vert \nu \vert^2 = 1$ and $0\leq\gamma\leq 1$.
The qubits 2 and 3 are prepared in a state, which is maximally entangled for $\mu=\nu=\frac{1}{\sqrt{2}}$ and is product for $\mu=1$ or $\nu=1$, and evolve freely. The sensors 1 and 4 are initially prepared in the Bell states $\frac{1}{\sqrt{2}}(\ket{1,1}+\ket{-1,-1})$ and the $X$-gate is applied to both qubits at time $t=\frac{1+\gamma}{2}T$. We choose this state as it is sufficiently complicated to illustrate all optimization steps, while still being relatively simple. In terms, of the state $\ket{\bm k}$ labeled by the $\bm k$ vectors the state is given by
\begin{equation}
    \ket{\psi} = \frac{\nu}{\sqrt 2} \ket{\bm k_1} +  \frac{\mu}{\sqrt 2} \ket{\bm k_2} + \frac{\mu}{\sqrt 2} \ket{\bm k_3} +  \frac{\nu}{\sqrt 2} \ket{\bm k_4}
\end{equation}
where $\bm k_1= - \bm k_3 =(\gamma,1,1,\gamma) $ and $\bm k_2 = -\bm k_4 = (\gamma,-1,-1,\gamma)$.

Notice that $\vert\psi\rangle$ is within the DFS and hence protected from the noise. Using (\ref{equ:K2}) we find its $\hat{K}$ matrix
\begin{align*}
    \hat{K} &= \begin{pmatrix}
        \gamma^2 & 0 & 0 & \gamma^2\\
        0 & 4\vert\mu\vert^2\vert\nu\vert^2 & 4\vert\mu\vert^2\vert\nu\vert^2 & 0\\
        0 & 4\vert\mu\vert^2\vert\nu\vert^2 & 4\vert\mu\vert^2\vert\nu\vert^2 & 0\\
        \gamma^2 & 0 & 0 & \gamma^2\\
    \end{pmatrix}
\intertext{and the QFIM }
    \mathcal{F} &= \frac{24^2T^2}{5^2} \begin{pmatrix}
        \gamma^2 + 4\vert\mu\vert^2\vert\nu\vert^2& \frac{\gamma^2}{2} + 4\vert\mu\vert^2\vert\nu\vert^2\\
        \frac{\gamma^2}{2}+4\vert\mu\vert^2\vert\nu\vert^2 & \frac{\gamma^2}{4}+4\vert\mu\vert^2\vert\nu\vert^2
    \end{pmatrix}.
\end{align*}

Following the \textbf{symmetrization}(1) step of the protocol discussed in \ref{sec: construction}, we obtain an improved state 
\begin{equation}
    \ket{\phi}= \frac{1}{\sqrt 2} \left(\ket{\text{GHZ}_{\bm k_1}}+\ket{\text{GHZ}_{\bm k_2}}\right)
\end{equation}
with amplitudes $\sqrt{\frac{\vert\mu\vert^2+\vert\nu\vert^2}{2}}=\frac{1}{\sqrt{2}}$. 

After the \textbf{sequentiallization}(2) step, we are left with the sequential GHZ strategy, where the two states $\ket{\text{GHZ}_{\bm k_1}}$ and $\ket{\text{GHZ}_{\bm k_2}}$ are prepared with equal probability. 

For the \textbf{vertices}(3) step we need to decompose the vectors $\bm k_1$ and $\bm k_2$ into vertices of the DFS. For our example there is a unique\footnote{The decomposition is unique up to exchanging $\bm v_1$ with $-\bm v_3$ and ${\bm v}_2$ with $-\bm v_4$, which does not change the final strategy.} decomposition given by
\begin{align*}
   \bm k_1 &= \frac{1+\gamma}{2}\bm v_1 + \frac{1-\gamma}{2} \bm v_2 \\
    \bm k_2 &= \frac{1-\gamma}{2}\bm v_3 + \frac{1+\gamma}{2} \bm v_4.
\end{align*}
Therefore we can use Lemma~\ref{lemma:sequential-vertex-GHZ} to get an improvement by using the sequential vertex protocol using $\ket{\text{GHZ}_{\bm v_1}}$ and $\ket{\text{GHZ}_{\bm v_2}}$ with equal probability $\frac{1}{2} = \frac{1}{2}(\frac{1+\gamma}{2}) +\frac{1}{2}(\frac{1-\gamma}{2})$. 
The performance of this strategy is characterized by 
\begin{align*}
    \hat{K} = \begin{pmatrix}
        1 & 0 & 0 & 1\\
        0 & 1 & 1 & 0\\
        0 & 1 & 1 & 0\\
        1 & 0 & 0 & 1\\
    \end{pmatrix}
\text{ and  }
    \mathcal{F} = \frac{24^2T^2}{5^2} \begin{pmatrix}
        2& \frac{3}{2} \\
        \frac{3}{2} & \frac{5}{4}
    \end{pmatrix}.
\end{align*}
This sequential vertex strategy is extremal and cannot be further improved.
Notice that the extremal strategy achieves the same performance as the initial state $\ket{\psi}$ for the choice $\mu=\nu=\frac{1}{\sqrt{2}}$ and $\gamma = 1$ and hence using Bell states in the bipartition (1,4) and (2,3) is extremal, too. This can additionally be seen by looking at the trace $\tr(\hat{K})=4=n$ which is maximal. Due to lemma~\ref{lemma:maximal-trace-extremal} is a sufficient criteria for extremal strategies (\ref{appendix:max_trace}).

\subsection{Optimal strategy}
Now we will compute the optimal strategy for the figure of merit $M= \tr(\mathrm{Cov} (\hat {\bm \alpha}))$. As we have shown this can be done by optimizing the rate with which different GHZ states pointing to vertices of the DFS are prepared. In our example there are only such states $\ket{\text{GHZ}_{\bm v_1}}$ and $\ket{\text{GHZ}_{\bm v_2}}$.
Introducing $\Delta = r_1-r_2$ the optimization in  (\ref{equ:opt}) simplifies to the following problem
\begin{equation}
    \min_\Delta \mathrm{tr}\left(\begin{pmatrix}
        1 & 0 \\
        0 & 1 
    \end{pmatrix}\left(
    \hat{S}
    \begin{pmatrix}
        1 & \Delta & \Delta & 1 \\
        \Delta & 1 & 1 & \Delta \\
        \Delta & 1 & 1 & \Delta \\
        1 & \Delta & \Delta & 1 \\
    \end{pmatrix}\hat{S}^\intercal
    \right)^{-1}\right)
\end{equation}
The solution is given by $M_{opt}=6.25$, and achieved by preparing $\ket{\text{GHZ}_{\bm v_1}}$ with rate $r_1=\frac{1}{6}$ and $\ket{\text{GHZ}_{\bm v_2}}$ with rate $r_2 = \frac{5}{6}$.

\section{Conclusion}\label{sec:conclusion}

The present study addressed the task of estimating the strength of multiple fields with different spatial patterns in the presence of strongly correlated noise. This is done by considering a distributed quantum sensor networks. Here, $n$ spin-$1/2$ sensors (qubits) at different locations are prepared in an initial entangled state and are coupled to the signals and the noise via the Pauli $\sigma_z^{j}$ operators. In addition, local control is performed during the evolution, which allows one to manipulate states with arbitrary effective spin values $\sigma_z^{j} \ket{k_j} = k_j \ket{k_j}$ with $|k_j|\leq 1$. We then build on the results of~\cite{sekatskiOptimalDistributedSensing2020, wolkNoisyDistributedSensing2020, hamannApproximateDecoherenceFree2022} and introduce the decoherence free subspace (DFS), where the states are decoupled from the noise. The DFS can be analyzed in terms of the spin vectors $\bm k = (k_1,k_2,\dots, k_n)$ labeling product state $\ket{\bm k}$, and defines a convex polytope inside the the associated vector space.

We consider multi-parameter estimations strategies that involve superpositions of states form the DFS and are thus completely protected from noise. We developed a systematic approach to compare such strategies and establish a pre-order emanating from the associated quantum Fisher Information matrix. We show that within the DFS all extremal QFIM are realised by vertex-sequential GHZ strategies -- strategies where only states of the form $\frac{1}{\sqrt 2}\left( \ket{\bm k} + \ket{-\bm k}\right)$, with the vector $\bm k$ pointing to a vertex of the DFS, are prepared in different runs of the experiment. We also discuss how to optimize over these strategies for a given figure of merit. Finally, a four-sensor example is provided to illustrate the application of our methods. We also demonstrate that among all noise-protected strategies the simple sequential GHZ strategies attain the optimal QFIM up to a factor of at most four.

This opens the way to use a simple generic single-parameter sensing strategy also for all multi-parameter estimation problems, and benefit from advantages such as noise protections against correlated noise, and readout by local measurements.

\section*{Acknowledgments}
A.H. and W.D. acknowledge support from the Austrian Science Fund (FWF) through the project P36009-N and P36010-N and Finanziert von der Europäischen Union.

\bibliography{ref.bib}

\appendix
\onecolumngrid
\section*{Appendix}


\section{Quantum Fisher information matrix}
\label{app: QFIM}

Consider a random variable $X_{\bm \alpha}$ distributed accordingly to the parametric model $p(x| \bm \alpha)$ with a vector parameter $\bm \alpha =(\alpha_1, \dots,\alpha_s)$. Around a point $\bm \alpha$  the FIM matrix is defined as
\begin{equation}
F_{\mu\nu}(X_{\bm \alpha}) = \sum_x \frac{\partial_\mu p(x|\bm \alpha) \partial_\nu p(x|\bm \alpha)}{ p(x|\bm \alpha)}
\end{equation}
with the derivatives taken with respect to different components of $\bm \alpha$ in the chosen frame. The FIM is real and symmetric $F^T=F$, furthermore it is positive semi-definite given that 
\begin{equation}
F(X_{\bm \alpha}) = \sum_x \left(\frac{\nabla p(x|\bm \alpha)}{\sqrt{p(x|\bm\alpha)}} \right) \left(\frac{\nabla p(x|\bm \alpha)}{\sqrt{p(x|\bm \alpha)}} \right)^T
\end{equation}
is a sum of terms proportional to 1-D projectors. As discussed in the main text the FIM is an important object in statistics, estimation theory and machine learning. In particular, in  estimation theory it is well known that covariance matrix $\text{Cov}(\hat {\bm \alpha})$ of an unbiased estimator $\hat {\bm \alpha}(X)$ satisfies
\begin{equation}
\text{Cov}(\hat {\bm \alpha}) \geq F^{-1}(X_{\bm \alpha})  
\end{equation}

Now consider a quantum parametric state model $\rho_{\bm \alpha}$ and define the QFIM $\cF(\rho_{\bm \alpha})$ as 
\begin{equation}
\cF_{\mu \nu}(\rho_{\bm \alpha}) = \tr\,  \rho_{\bm \alpha} \frac{1}{2}\{L_\mu,L_\nu\}   = \tr\, \partial_\nu \rho_{\bm \alpha} L_\mu = \tr\,  \partial_\mu \rho_{\bm \alpha} L_\nu
\end{equation}
with the SLD operators defined as solutions of  
\begin{equation}\label{app a: SLD}
\frac{1}{2}\{\rho_{\bm \alpha}, L_\mu\} = \partial_\mu \rho_{\bm \alpha}.
\end{equation} 
The derivative of the density matrix with respect to a parameter defines an Hermitian operator  $\partial_\mu \rho_{\bm \alpha} =(\partial_\mu \rho_{\bm \alpha})^\dag$, because $\rho_{\bm \alpha} $ and $\rho_{\bm \alpha}  + \partial_\mu \rho_{\bm \alpha} \dd \alpha_\mu$ are Hermitian. Hence, the equation $\frac{1}{2}\{\rho_{\bm \alpha}, L_\mu\} = \partial_\mu \rho_{\bm \alpha}$ implies that the SLD opeartors are also Hermitian $L_\mu = L_\mu^\dag$ (or at least can be taken to be Hermitian since $L_\mu' =\frac{1}{2} (L_\mu +L_\mu^\dag)$ solves the SLD equation if $L_\mu$ does).

\subsection{The QFIM upper bound the FIM for all measurements.}

\textbf{Observation.} \textit{The QFIM upper bounds the FIM }
\begin{equation}\label{app a: F>cF}
F(X_{\bm \alpha}) \leq  \cF (\rho_{\bm \alpha})
\end{equation}
\textit{obtained for any POVM $\{E_x\}_x$ mapping $\rho_{\bm \alpha}$ to the parametric model $p(x|\bm \alpha) = \tr \rho_{\bm \alpha} E_x$. } 

\textit{Proof:} For completeness we give a proof of this well known fact below.  We will sometimes omit the subscript $\bm \alpha$ below to lighten the equations. 
For and POVM $\{E_x\}_x$ we can write down the FIM explicitly
\begin{equation}
    F_{\mu \nu}(X_{\bm \alpha}) = \sum_x \frac{(\tr E_x \partial_\mu \rho) (\tr E_x \partial_\nu \rho)}{\tr E_x\rho}.
\end{equation} Now, let us take any real vector $\bm v = (v_1,\dots, v_s)$ and compute $F_{(\bm v)} = \bm v F(X_{\bm \alpha})\bm v^T$, we find
\begin{equation} \label{app: 1 some eq}
F_{(\bm v)} = \sum_{\nu \mu} v_\nu v_\mu \cF_{\mu \nu}(X_{\bm \alpha}) = \sum_x \frac{1}{\tr E_x \rho }( \tr E_x \sum_{\mu} v_\mu \partial_\mu \rho)^2.
\end{equation}
Using $ \partial_\mu \rho = \frac{1}{2}\{\rho, L_\mu\}$ observe that for an Hermitian operator $E_x$ one has
\begin{equation}
\begin{split}
    \tr E_x \sum_{\mu} v_\mu \partial_\mu \rho &= \frac{1}{2} \tr E_x \{\rho, \sum_\mu v_\mu L_\mu \}
= \frac{1}{2} \tr E_x \rho (\sum_\mu v_\mu L_\mu) + \frac{1}{2}  \tr E _x^\dag (\sum_\mu v_\mu L_\mu)^\dag \rho^\dag \\
& = \frac{1}{2} \tr E_x \rho (\sum_\mu v_\mu L_\mu) + \frac{1}{2}  \tr (E_x \rho (\sum_\mu v_\mu L_\mu) )^\dag = \frac{1}{2} \tr E_x \rho (\sum_\mu v_\mu L_\mu) + \frac{1}{2}  \left(\tr E_x \rho (\sum_\mu v_\mu L_\mu) \right)^* \\
& = \text{Re}\left[ \tr \rho E_x \sum_\mu v_\mu L_\mu\right]
 \end{split}
\end{equation}
where we used the fact that $L_\mu$ and $\rho$ are Hermitian in the second equality. Plugging this identity into Eq.~\eqref{app: 1 some eq} gives
\begin{equation} \label{app a: some else eq}
F_{(\bm v)} =  \sum_x \frac{1}{\tr E_x \rho } \left(\text{Re}[ \tr \rho E_x \sum_\mu v_\mu L_\mu]\right)^2 \leq \sum_x \frac{1}{\tr E_x \rho } \left| \tr \rho E_x \sum_\mu v_\mu L_\mu\right|^2 
\end{equation}
Using the Cauchy-Schwartz inequality for the Hilber-Schmidt inner product ($|\tr AB^\dag |^2 \leq \tr A A^\dag \tr B B^\dag$) we can get further bound $\left| \tr \rho E_x (\sum_\mu v_\mu L_\mu)\right|^2 = \left| \tr \sqrt{ \rho E_x}  \sqrt{E_x} (\sum_\mu v_\mu L_\mu) \sqrt{\rho} \right|^2 \leq \tr E_x \rho \, \tr \rho (\sum_\mu v_\mu L_\mu) E_x (\sum_\nu v_\nu L_\nu)$. Plugging in Eq.~\eqref{app a: some else eq} gives
\begin{equation} \label{app a: eq whatever}
\begin{split}
 F_{(\bm v)}&
\leq \sum_x \frac{1}{\tr E_x \rho } \tr E_x \rho\left(\tr \rho (\sum_\mu v_\mu L_\mu) E_x (\sum_\nu v_\nu L_\nu)\right) = \sum_x  \tr \rho (\sum_\mu v_\mu L_\mu) E_x (\sum_\nu v_\nu L_\nu) \\
& = \tr \rho (\sum_\mu v_\mu L_\mu)(\sum_\nu v_\nu L_\nu)
= \bm v \cF(\rho_{\bm \alpha}) \bm v^T
\end{split} 
\end{equation}
showing that $\bm v F(X_{\bm \alpha}) \bm v^T \leq \bm v \cF(\rho_{\bm \alpha}) \bm v^T$ for all $\bm v$ and hence $ F(X_{\bm \alpha}) \leq  \cF(\rho_{\bm \alpha}) \square$.\\

It is worth noting that the last line of Eq.~\eqref{app a: eq whatever} makes it explicit that the term  $\bm v \cF(\rho_{\bm \alpha}) \bm v^T= \tr \rho (\sum_\mu v_\mu L_\mu)^2\geq 0$ is the expected value of a non-negative operator and thus non-negative, showing that the QFIM is positive semi-definite (it is Hermitian by construction). 

The inequality $F(X_{\bm \alpha}) \leq  \cF(\rho_{\bm \alpha}) $ directly leads to the quantum Cramer-Rao bound, as for any two positive operators such that $F(X_{\bm \alpha}) \leq  \cF(\rho_{\bm \alpha})$ one automatically has $F^{-1}(X_{\bm \alpha}) \geq \cF^{-1}(\rho_{\bm \alpha})$. Hence the following bound holds
\begin{equation}
\text{Cov}(\hat {\bm \alpha}) \geq \cF^{-1}(\rho_{\bm \alpha})
\end{equation}
assuming that the inverse exists.

\subsection{Tightness of the QFIM bound on FIM }

Let us briefly comment on the tightness of Eq.~\eqref{app a: F>cF}. In the derivation we used several inequalities that become saturated if
\begin{equation}
    \left(\text{Re}[ \tr \rho E_x \sum_\mu v_\mu L_\mu]\right)^2 \overset{!}{=} \tr E_x \rho \, \tr \rho (\sum_\mu v_\mu L_\mu) E_x (\sum_\nu v_\nu L_\nu)
\end{equation}
for all $x$. Denoting $L_{(\bm v)} = \sum_\mu v_\mu L_\mu$ gives a shorter condition
\begin{equation}
    \left(\text{Re}[ \tr \rho E_x L_{(\bm v)}]\right)^2 \overset{!}{=} \tr E_x \rho \, \tr \rho L_{(\bm v)} E_x L_{(\bm v)}.
\end{equation}
It is easy to see that the equality is  satisfied by setting $E_x = \Pi_x$, where $\Pi_x$ project on the eigespaces of $L_{(\bm v)} = \sum_x \ell_x \Pi_x$. If all the SLDs commute $[L_\mu, L_\nu]=0$ this measurement choice is compatible with all directions $\bm v F(X_{\bm \alpha}) \bm v^T = \bm v \cF(\rho_{\bm \alpha}) \bm v^T$, and leads to $F(X_{\bm \alpha}) =  \cF (\rho_{\bm \alpha})$. However, in general  a measurement saturating $F(X_{\bm \alpha}) \geq \cF (\rho_{\bm \alpha})$ does not  exist.\\

\subsection{Attainability of the QFIM with local measurement for GHZ states}
\label{app: GHZ local meas}

Let us now focus on the specific case of our interest where the encoding is unitary $\rho_{\bm \alpha}=\ketbra{\psi_{\bm \alpha}}$ with $ \ket{\psi_{\bm \alpha}} = U_{\bm \alpha} \ket{\psi}$, where 
    $U_{\bm\alpha} = \exp\left( - \ii\, T \sum_\mu \alpha_\mu G_\mu \right)$
and the generators  commute  $[G_\mu, G_\nu]=0$. The derivatives of the state thus read
\begin{equation}\label{app a: derivative unitary}
\partial_\mu \rho_{\bm \alpha} = - \ii \, T [G_\mu, \rho_{\bm \alpha}] = T \, U_{\bm \alpha} \left( \ketbra{\partial_\mu \psi}{\psi} +  \ketbra{\psi}{\partial_\mu \psi}\right) U_{\bm \alpha}^\dag 
\end{equation}
with $\ket{\partial_\mu \psi} = - \ii G_\mu \ket{\psi}$. Hence the SLD equation~\eqref{app a: SLD} becomes
\begin{equation}\label{app a: SLD 2}
    \ketbra{\partial_\mu \psi}{\psi} +  \ketbra{\psi}{\partial_\mu \psi} = \frac{1}{2 T} \left( \ketbra{\psi} \tilde L_\mu + \tilde L_\mu  \ketbra{\psi} \right) 
\end{equation}
for $\tilde L_\mu = U_{\bm \alpha}^\dag L_\mu U_{\bm \alpha}$.
This is equation is straightforward to solve by defining component of $\ket{\partial_\mu \psi}$ orthogonal to $\ket{\psi}$  as 
\begin{equation}
    (\id - \ketbra{\psi}) \ket{\partial_\mu \psi} = - \ii (\id - \ketbra{\psi}) G_\mu \ket{\psi} = \beta_\mu \ket{\psi_\mu^\perp}
\end{equation} where $\beta_\mu=\sqrt{\bra{\psi}G_\mu (\id - \ketbra{\psi}) G_\mu \ket{\psi}}$ is a positive scalar. 
It is straightforward then to verify that Eq.~\eqref{app a: SLD 2} is solved by the operator
\begin{equation} \label{eq: SLD app}    
\tilde L_\mu = 2  \ii\,  T \beta_\mu \left(\ketbra{\psi}{\psi_\mu^\perp}- \ketbra{\psi_\mu^\perp}{\psi} \right).
\end{equation}
For the QFIM this yields
\begin{equation}\label{appeq: QFI pure}
    \cF_{\mu \nu}(\rho_{\bm \alpha})  = \frac{1}{2} \tr \rho_{\bm \alpha} \{L_\mu, L_\nu\} = \frac{1}{2} \bra{\psi}  \{\tilde L_\mu, \tilde L_\nu\} \ket{\psi}
    = 4 T^2 \text{Re}\left[ \bra{\psi} G_\mu G_\nu \ket{\psi} - \bra{\psi} G_\mu \ket{\psi} \bra{\psi} G_\nu \ket{\psi} \right].
\end{equation}

Notably, even though the generators commute $[G_\mu, G_\nu]$, this is not necessarily the case for the SLDs $[\tilde L_\mu, \tilde L_\nu ]\neq 0$ as one can see from Eq.~\eqref{eq: SLD app}. 

Nevertheless, if the initial state is of GHZ type, i.e. a superposition of two eigenstates of the generators e.g. $\ket{\psi} = \frac{1}{\sqrt{2}}\left( \ket{\bm k}+ \ket{-\bm k}\right)$, all the SLD operators are proportional to
\begin{equation}
\tilde L_\mu \propto \sigma_Y^L = \ii \ketbra{-\bm k}{\bm k} - \ii \ketbra{- \bm k}{\bm k}.
\end{equation}
This follows from the fact $\ket{\psi_\mu^\perp} = \ket{\psi^\perp} \frac{1}{\sqrt{2}}\left( \ket{\bm k} -  \ket{-\bm k} \right)$ for all generators,
since $\ket{\psi^\perp}$ is  the unique state orthogonal to $\ket{\psi}$ in the subspace $\mathcal{H}_L$ spanned by $\ket{\bm k}$ and $\ket{- \bm k}$. Furthermore, since the state $\ket{\bm k}$ and $\ket{- \bm k}$ are product over different sensors, on the subspace $\mathcal{H}_L$ the observable $\sigma_Y^L$ can be measured locally. This is done by measuring $\ii \ketbra{-k_j}{k_j} - \ii \ketbra{k_j}{-k_j}$ on each sensor and combining the outcomes.


\subsection{QFIM for pure states within the DFS}\label{appendix:QFIM}

Let us now compute the QFIM elements in Eq. \eqref{appeq: QFI pure}
\begin{equation}
 \mathcal{F}_{ij} = 4 T^2 \mathrm{Re}(\underbrace{\langle\psi\vert G_i^\mathrm{eff} G_j^\mathrm{eff}\left\vert\psi\right\rangle}_{\mathrm{I}_{ij}})- 4 T^2 \underbrace{\langle\psi\vert G_i^\mathrm{eff}\left\vert\psi\right\rangle}_{\mathrm{II}_{i}}\underbrace{\langle\psi\vert G_j^\mathrm{eff}\left\vert\psi\right\rangle}_{\mathrm{II}_{j}}.
\end{equation}
with the help of the vector representation of eigenvalues introduced in the main text $G_i^\mathrm{eff}\left\vert\bm k\right\rangle  = \boldsymbol{e}_i^\intercal\hat{S} \boldsymbol{k} \vert \bm k \rangle$.
Note that the eigenvalues real, and for the term $\mathrm{I}_{ij}$ identified in the last equation we then get
\begin{equation}
 \mathrm{I}_{ij} = \sum_{\boldsymbol{k} \in \mathrm{kernel}(\hat{N})} \left\vert c_{\boldsymbol{k}}\right\vert^2(\boldsymbol{e}_i^\intercal\hat{S} \boldsymbol{k}) (\boldsymbol{k}^\intercal \hat{S}^\intercal \boldsymbol{e}_j) = \boldsymbol{e}_i^\intercal\hat{S} \left( \sum_{\boldsymbol{k} \in \mathrm{kernel}(\hat{N})} \left\vert c_{\boldsymbol{k}}\right\vert^2 \boldsymbol{k} \boldsymbol{k}^\intercal\right)\hat{S}^\intercal \boldsymbol{e}_j.
\end{equation}
The other term reads
 \begin{equation}    \mathrm{II}_i= \left( \sum_{\boldsymbol{k} \in \mathrm{kernel}(\hat{N})} \left\vert c_{\boldsymbol{k}}\right\vert^2 \boldsymbol{e}_i^\intercal\hat{S} \boldsymbol{k} \right),
\end{equation}
 and by introducing $\boldsymbol{\bar{k}}= \sum_{\boldsymbol{k} \in \mathrm{kernel}(\hat{N})} \left\vert c_{\boldsymbol{k}}\right\vert^2 \boldsymbol{k}$ it can be cast into a simpler form
 \begin{equation}    \mathrm{II}_i= \boldsymbol{e}_i^\intercal\hat{S} \boldsymbol{\bar{k}}.
\end{equation}

Finally, noting that $\mathrm{II}_i \mathrm{II}_j =\boldsymbol{e}_i^\intercal   \hat{S} (\boldsymbol{\bar{k}}\boldsymbol{\bar{k}}^\intercal ) \hat{S}^\intercal \boldsymbol{e}_j$, and introducing the matrix
\begin{equation}
 \hat{K} = \sum_{\boldsymbol{k} \in \mathrm{kernel}(\hat{N})} \left\vert c_{\boldsymbol{k}}\right\vert^2 \boldsymbol{k} \boldsymbol{k}^\intercal - \boldsymbol{\bar{k}}\boldsymbol{\bar{k}}^\intercal
\end{equation}
we get for the QFIM elements
\begin{align}
 \mathcal{F}_{ij}&= 4 T^2 (\mathrm{I}_{ij} -\mathrm{II}_i \mathrm{II}_j )
 = 4 T^2 \boldsymbol{e}_i^\intercal\hat{S} \left( \sum_{\boldsymbol{k} \in \mathrm{kernel}(\hat{N})} \left\vert c_{\boldsymbol{k}}\right\vert^2 \boldsymbol{k} \boldsymbol{k}^\intercal - \boldsymbol{\bar{k}}\boldsymbol{\bar{k}}^\intercal\right)\hat{S}^\intercal \boldsymbol{e}_j \\
 &= 4 T^2 \boldsymbol{e}_i^\intercal\hat{S} 
 \hat{K} \hat{S}^\intercal \boldsymbol{e}_j.
\end{align}
The QFIM of our bit-flip assisted strategies is thus given by a compact expression
\begin{equation}
     \mathcal{F} = 4T^2 \hat{S} \hat{K} \hat{S}^\intercal \label{equ:QFIM},
\end{equation}
where $\hat{K}$ is a real positive matrix that is very helpful to compare different strategies.

\subsection{QFIM for states not in the DFS}\label{appendix:QFIM-not-in-dfs}
Notice that final state 
\begin{align}
\rho(\boldsymbol{\alpha})&= \mathcal{E}_{\boldsymbol{\alpha}}(\mathcal{D}(\rho_{\mathrm{init}}))
\intertext{is given by first applying the noise}
    \mathcal{D}(\rho) &=  \int_{\boldsymbol{\beta}} p(\boldsymbol{\beta}) U_{\boldsymbol{\beta}}\rho U_{\boldsymbol{\beta}}^\dagger\\
    U_{\boldsymbol{\beta}} &= \exp\left(-iT \sum_{i} \beta_i G_i^{\text{noise}} \right)\\
    \intertext{and then the parameter encoding}
    \mathcal{E}(\rho) &= U_{\boldsymbol{\alpha}} \rho U_{\boldsymbol{\alpha}}^\dagger \\
    U_{\boldsymbol{\alpha}} &= \exp\left(-iT \sum_{i} \alpha_i G_i \right).
\end{align}
Therefore the QFIM for a state not in the DFS can be computed by computing the QFIM for the mixed state $\mathcal{D}(\rho_{\mathrm{init}})$.

\subsubsection{QFIM for mixed states}\label{appendix:khat-mixed}
\begin{align}
    \intertext{Given a mixed state $\rho=\sum_{i,j} c_{i,j} \left\vert \boldsymbol{k}_i \right\rangle \langle \boldsymbol{k}_j \vert$ with a spectral decomposition of the components matrix $c = U P U^\dagger$ and $P=\mathrm{diag}(p_1,p_2,...)$ being a diagonal matrix. The spectral decomposition is given by}
    \rho &= \sum_i p_i  \left\vert i \right\rangle \langle i \vert = \sum_{i, \nu,\eta} U_{\nu,i} p_i U^*_{\eta,i}  \left\vert \boldsymbol{k}_\nu \right\rangle \langle \boldsymbol{k}_\eta \vert\\
    \left\vert i \right\rangle &= \sum_{\nu} U_{\nu,i} \left\vert \boldsymbol{k}_\nu \right\rangle\\
    \intertext{Then the QFIM components for the parameters encoded by $G_a$ and $G_b$ can be computed using the equation from~\cite{LiuYuanLuWang2019} }
    \mathcal{F}_{a,b} &= 2 \sum_{i,j} \frac{(p_i-p_j)^2}{p_i+p_j} \langle i \vert G_a \left\vert j\right\rangle\langle j \vert G_b \left\vert i\right\rangle\\
    \langle i \vert G_a \left\vert j\right\rangle &= \langle i \vert G_a \sum_{\nu} U_{\nu,j} \left\vert \boldsymbol{k}_\nu\right\rangle \\
    &= \langle i \vert \sum_\nu T \boldsymbol{e}_a^\intercal \hat{S} \boldsymbol{k}_\nu U_{\nu,j}\left\vert \boldsymbol{k}_\nu \right\rangle\\
    &= T \sum_{\nu,\eta} \langle \boldsymbol{k}_\eta \vert U^*_{\eta,i} \boldsymbol{e}_a^\intercal \hat{S} \boldsymbol{k}_\nu U_{\nu,j}\left\vert \boldsymbol{k}_\nu \right\rangle\\
    &= T \sum_{\nu,\eta} U^*_{\eta,i} \boldsymbol{e}_a^\intercal \hat{S} \boldsymbol{k}_\nu U_{\nu,j}\vert   \langle \boldsymbol{k}_\eta \left\vert  \boldsymbol{k}_\nu \right\rangle \\
    &=  T \boldsymbol{e}_a^\intercal \hat{S} \sum_{\nu}  U^*_{\nu,i} U_{\nu,j} \boldsymbol{k}_\nu \\
    \mathcal{F}_{a,b} &= 2 \sum_{i,j,\nu,\eta} \frac{(p_i-p_j)^2}{p_i+p_j}   \underbrace{T \boldsymbol{e}_a^\intercal \hat{S} U^*_{\nu,i} U_{\nu,j} \boldsymbol{k}_\nu}_{\langle i \vert G_a \left\vert j\right\rangle} 
    \underbrace{T U^*_{\eta,i} U_{\eta,j} \boldsymbol{k}_\eta^\intercal  \hat{S}^\intercal\boldsymbol{e}_b}_{\langle j \vert G_b \left\vert i\right\rangle}\\
    &= 2 T^2 \boldsymbol{e}_a^\intercal \hat{S}  \underbrace{\sum_{i,j,\nu,\eta} \frac{(p_i-p_j)^2}{p_i+p_j}    U^*_{\nu,i} U_{\nu,j}  
    U^*_{\eta,i} U_{\eta,j} \boldsymbol{k}_\nu \boldsymbol{k}_\eta^\intercal}_{2 \hat{K}}  \hat{S}^\intercal\boldsymbol{e}_b\\
    \hat{K} &= \frac{1}{2}\sum_{i,j,\nu,\eta} \frac{(p_i-p_j)^2}{p_i+p_j} U^*_{\nu,i} U_{\nu,j} U^*_{\eta,i} U_{\eta,j} \boldsymbol{k}_\nu \boldsymbol{k}_\eta^\intercal.
\end{align}

\section{Infinitely strong noise}
\label{sec:infiniteNoise}
Given any pure initial state $\left\vert \psi \right\rangle = \sum_{\boldsymbol{k}} c_{\boldsymbol{k}} \left\vert \boldsymbol{k}\right\rangle$ and noise fields fluctuation infinitely strong and independently, the noise 
leads to a state  
\begin{align}
    \rho &= \mathcal{D}(\left\vert
    \psi\right\rangle)\\
    &= \bigoplus_{\boldsymbol{\kappa}} P_{\mathrm{DFS}_{\bm \kappa}} \left\vert \psi \right\rangle\langle \psi \vert P_{\mathrm{DFS}_{\bm \kappa}} ,
    \intertext{which is block diagonal with respect to the different affine DFS. Therefore the QFIM reduces to the QFIM within each affine DFS.}
    \intertext{This can be seen by writing $\rho$ in the $\left\vert\boldsymbol{k}\right\rangle$ basis}
    \rho_{\boldsymbol{k},\boldsymbol{k'}} &= \langle \boldsymbol{k}\vert \int_{\boldsymbol{\beta}} p(\boldsymbol{\beta}) U_{\boldsymbol{\beta}} \left\vert \psi \right\rangle\langle\psi \vert U^\dagger d\boldsymbol{\beta}\left\vert\boldsymbol{k}'\right\rangle\\
    &= \int_{\boldsymbol{\beta}}  p(\boldsymbol{\beta}) c_{\boldsymbol{k}}c^*_{\boldsymbol{k'}} \exp{\left(-i T \boldsymbol{\beta}^\intercal \hat{N} \boldsymbol{k}\right)} \exp{\left(i T \boldsymbol{\beta}^\intercal \hat{N} \boldsymbol{k'}\right)} d\boldsymbol{\beta}\\
    &= \int_{\boldsymbol{\beta}}  p(\boldsymbol{\beta}) c_{\boldsymbol{k}}c^*_{\boldsymbol{k'}} \exp{\left(-i T \boldsymbol{\beta}^\intercal \hat{N} (\boldsymbol{k}-\boldsymbol{k'})\right)} d\boldsymbol{\beta}\\
    &= c_{\boldsymbol{k}}c^*_{\boldsymbol{k'}} \int_{\boldsymbol{\beta}}  p(\boldsymbol{\beta})  \exp{\left(-i T \boldsymbol{\beta}^\intercal \hat{N} (\boldsymbol{k}-\boldsymbol{k'})\right)} d\boldsymbol{\beta}\\
    \intertext{As by assumption the noise fluctuations are infinitely strong, the integral is only non-zero if the exponent is zero.  }
    &=  \begin{cases}
        c_{\boldsymbol{k}}c^*_{\boldsymbol{k'}} &\text{ if } \mathrm{DFS}_\kappa,\\
        0 &\text{ else }
    \end{cases}
\end{align}


\section{Construction extremal strategy}\label{appendix:extremal}
\subsection{Step 1: Symmetrization}\label{appendix:symmetrization}
Notice that the DFS always contains pairs of $(\boldsymbol{k},\boldsymbol{-k})$, which map to the same GHZ state. In order to get rid of this redundancy we introduce the label set $\kappa_+\subset \kappa$, which for $\boldsymbol{k}\in \kappa$ either contains $\boldsymbol{k}$ or $-\boldsymbol{k}$, but not both. Therefore $\kappa_+$ uniquely labels the required GHZ states. 
There is a rather mathematical special case if $\kappa$ contains the zero vector. The corresponding GHZ state is not well defined, which can be fixed, but is then anyhow not useful for sensing. Therefore the zero vector is not included into $\kappa_+$. 
\begin{lemma}\label{lemma:superpositon-GHZ}
    Given an arbitrary strategy $\left\vert\psi\right\rangle = \sum_{\boldsymbol{k}\in\kappa\subset\mathrm{DFS}} c'_{\boldsymbol{k}}\left\vert\boldsymbol{k}\right\rangle$ within the DFS, then there exits a superposition of GHZ states
    $$\left\vert \phi \right\rangle = \sum_{\boldsymbol{k}\in\kappa_+\subset\mathrm{DFS}} \sqrt{|c'_{\boldsymbol{k}}|^2+|c'_{\boldsymbol{-k}}|^2}\left\vert \mathrm{GHZ}_{\boldsymbol{k}}\right\rangle,$$ 
    such that $\mathcal{F}_{\phi}\geq \mathcal{F}_{\psi}$, i.e. the QFIM is not worse.
\end{lemma}

\begin{proof}\label{proof:superposition-GHZ} First, we notice that $\left\vert\phi\right\rangle$ is a valid normalised state. 

Secondly, we show that the QFIM is not worse:
\begin{align}
    \mathcal{F}_\phi&\geq \mathcal{F}_\psi \iff  \mathcal{F}_\phi - \mathcal{F}_\psi \geq 0 \\
    & \iff S \left( \hat{K}_{\phi} - \hat{K}_{\psi}\right)  S^\intercal \geq 0 \\
    \intertext{Therefore comparing the $\hat{K}$ matrices is sufficient}
    \intertext{Lets start by writing $\left\vert\phi\right\rangle$ in the $\left\vert\boldsymbol{k}\right\rangle$ basis}
    \left\vert \phi \right\rangle &= \sum_{\boldsymbol{k}\in\kappa} \frac{\sqrt{|c_{\boldsymbol{k}}|^2+|c_{\boldsymbol{-k}}|^2}}{\sqrt{2}}\left\vert \boldsymbol{k}\right\rangle\\
    \bar{k}_{\phi} &= \sum_{\boldsymbol{k}} \frac{\left\vert c_{\boldsymbol{k}}\right\vert^2 + \left\vert c_{-\boldsymbol{k}}\right\vert^2}{2} \boldsymbol{k} = 0\\
    \hat{K}_{\phi} &= \sum_{\boldsymbol{k}} \frac{\left\vert c_{\boldsymbol{k}}\right\vert^2 + \left\vert c_{-\boldsymbol{k}}\right\vert^2}{2} \boldsymbol{k}\boldsymbol{k}^\intercal\\
    \hat{K}_{\psi} &= \underbrace{\sum_{\boldsymbol{k}} \left\vert c_{\boldsymbol{k}}\right\vert^2  \boldsymbol{k}\boldsymbol{k}^\intercal}_{I} - \bar{k}_{\psi}\bar{k}_{\psi}^\intercal
    \intertext{Notice that $\boldsymbol{k}\boldsymbol{k}^\intercal=(\boldsymbol{-k})(\boldsymbol{-k})^\intercal$}
    I &= \sum_{\boldsymbol{k}} \left\vert c_{\boldsymbol{k}}\right\vert^2  \left(\frac{\boldsymbol{k}\boldsymbol{k}^\intercal}{2}+\frac{(\boldsymbol{-k})(\boldsymbol{-k})^\intercal}{2}\right)\\
    &= \sum_{\boldsymbol{k}} \left\vert c_{\boldsymbol{k}}\right\vert^2  \frac{\boldsymbol{k}\boldsymbol{k}^\intercal}{2}+\sum_{\boldsymbol{k}} \left\vert c_{\boldsymbol{k}}\right\vert^2\frac{(\boldsymbol{-k})(\boldsymbol{-k})^\intercal}{2}\\
    &= \sum_{\boldsymbol{k}} \left\vert c_{\boldsymbol{k}}\right\vert^2  \frac{\boldsymbol{k}\boldsymbol{k}^\intercal}{2}+\sum_{\boldsymbol{k}} \left\vert c_{\boldsymbol{-k}}\right\vert^2\frac{\boldsymbol{k}\boldsymbol{k}^\intercal}{2}\\
    &= \sum_{\boldsymbol{k}} \frac{\left\vert c_{\boldsymbol{k}}\right\vert^2+ \left\vert c_{\boldsymbol{-k}}\right\vert^2}{2} \boldsymbol{k}\boldsymbol{k}^\intercal = \hat{K}_{\phi}
    \intertext{Which finally leads to}
    \hat{K}_{\phi} - \hat{K}_{\psi} &= \bar{k}_{\psi}\bar{k}_{\psi}^\intercal \geq 0 
\end{align}
\end{proof}

\subsection{Step 2: Sequentiallization}\label{appendix:sequential}
\begin{lemma}\label{lemma:sequential-GHZ}
        Given a superposition of GHZ states 
        $$\left\vert \phi \right\rangle = \sum_{\boldsymbol{k}\in\kappa_+} c_{\boldsymbol{k}} \left\vert \mathrm{GHZ}_{\boldsymbol{k}} \right\rangle$$
        then the QFIM can be achieved by individually/sequentially probing $\left\vert \mathrm{GHZ}_{\boldsymbol{k}}\right\rangle$ with probability $\left\vert c_{\boldsymbol{k}}\right\vert^2$.
\end{lemma}
\begin{proof}\label{proof:sequential-GHZ} The initial state of a sequential protocol can be written as a direct sum $\bigoplus_i p_i \rho_i$
The QFIM in this case is given as a sum over the blocks.
\begin{align}
    \mathcal{F}_{\mathrm{seq}} &= \sum_i \mathcal{F}_{\rho_i}
    \intertext{Hence for the GHZ states we get}
    &= 4 T^2 \sum_{\boldsymbol{k}} \left\vert c_{\boldsymbol{k}}\right\vert^2 S \boldsymbol{k}\boldsymbol{k}^\intercal S^\intercal = 4 T^2 S  \sum_{\boldsymbol{k}} \left\vert c_{\boldsymbol{k}}\right\vert^2 \boldsymbol{k}\boldsymbol{k}^\intercal S^\intercal
    \intertext{For the superposition of GHZ states we get in the $\left\vert\boldsymbol{k}\right\rangle$ basis}
    \left\vert\phi\right\rangle &= \sum_{\boldsymbol{k}} \frac{c_{\boldsymbol{k}}}{\sqrt{2}} \left\vert\boldsymbol{k}\right\rangle + \frac{c_{\boldsymbol{k}}}{\sqrt{2}}\left\vert-\boldsymbol{k}\right\rangle
    \intertext{Therefore using (\ref{equ:QFIM})}
    \mathcal{F}_{\mathrm{sup}} &= 4T^2 S\sum_{\boldsymbol{k}} \left\vert c_{\boldsymbol{k}}\right\vert^2 \frac{\boldsymbol{k}\boldsymbol{k}^\intercal+(-\boldsymbol{k})(-\boldsymbol{k})^\intercal)}{2}S^\intercal\\
    &=\mathcal{F}_{\mathrm{seq}}
\end{align}
\end{proof}

\subsection{Step 3: Vertices}\label{appendix:vertices}
\begin{lemma}\label{lemma:sequential-vertex-GHZ}
    Given a vector $\boldsymbol{k}\in DFS$ and a convex decomposition $\boldsymbol{k}=\sum_i p_i \boldsymbol{v}_i$ into vertices $\boldsymbol{v}_i$ of the DFS.
    Let $\left\vert \psi\right\rangle$ be the into $\boldsymbol{k}$ directed GHZ state and $\left\vert \phi \right\rangle = \bigoplus_i \sqrt{p_i} \left\vert\mathrm{GHZ}_{\boldsymbol{v}_i}\right\rangle$ be the corresponding sequential vertex GHZ strategy.
    Then the QFIM of the GHZ is improved or equal
    $$K_{\psi}\leq K_\phi \text{ and }\mathcal{F}_{\psi}\leq \mathcal{F}_\phi$$ by the sequential vertex GHZ strategy.
\end{lemma}
\begin{proof}\label{proof:sequential-vertex-GHZ} Notice that, the DFS is a convex set in $d$ dimensions. 
Due to Carathéodory's Theorem   
\begin{align}
    \boldsymbol{k} &= \sum_{i=1}^{d+1} p_i \boldsymbol{v}_i
    \intertext{every point in this polytope can be decomposed into a convex combination of $d+1$ vertices. }
    \intertext{Therefore it remains to show that}
    K_{\psi}&\leq K_\phi \\
    \Leftrightarrow 0 &\leq \sum_i p_i \boldsymbol{v}_i\boldsymbol{v}_i^\intercal - \boldsymbol{k}\boldsymbol{k}^\intercal  \\
    &= \sum_i p_i \boldsymbol{v}_i\boldsymbol{v}_i^\intercal - \boldsymbol{k}\boldsymbol{k}^\intercal  - \boldsymbol{k}\boldsymbol{k}^\intercal  + \boldsymbol{k}\boldsymbol{k}^\intercal \\
    &= \sum_i p_i \boldsymbol{v}_i\boldsymbol{v}_i^\intercal  - \sum_i p_i \boldsymbol{v}_i\boldsymbol{k}^\intercal  - \sum_i p_i \boldsymbol{k}\boldsymbol{v}_i^\intercal + \sum_i p_i \boldsymbol{k}\boldsymbol{k}^\intercal \\
    &= \sum_i p_i (\boldsymbol{v}_i-\boldsymbol{k})(\boldsymbol{v}_i-\boldsymbol{k})^\intercal
\end{align}
Notice that convex sum of positive operators are positive and projectors are positive operators.
\end{proof}

\subsection{Sequential vertex GHZ strategies are extremal}\label{appendix:vertex-GHZ-extremal}
\begin{lemma}\label{lemma:vertex-GHZ-extremal}
    Sequential vertex GHZ strategy $$\left\vert \psi \right\rangle = \bigoplus_i \sqrt{p_i} \left\vert\mathrm{GHZ}_{\boldsymbol{v}_i}\right\rangle$$ are extremal.
\end{lemma}
\begin{proof}\label{proof:vertex-GHZ-extremal} The Proof of Lemma~\ref{lemma:vertex-GHZ-extremal}

 We start with 
\begin{align}
    \left\vert \phi \right\rangle &= \bigoplus_i \sqrt{p_i} \left\vert\mathrm{GHZ}_{\boldsymbol{v}_i}\right\rangle
    \intertext{a being a vertex sequential strategy where we assume w.l.o.g that $\boldsymbol{v}_i$ are unique meaning  $i\neq j \implies \boldsymbol{v}_i\neq\boldsymbol{v}_j$  and that $\boldsymbol{v}_i \neq -\boldsymbol{v}_j$ }
    \left\vert \psi \right\rangle &= \bigoplus_j \sqrt{\Tilde{p}_j} \left\vert\mathrm{GHZ}_{\boldsymbol{k}_j}\right\rangle
    \intertext{being any sequential GHZ strategy w.l.o.g we can assume that $\boldsymbol{k}_i$ are unique in the same way. Additionally we require}
    \hat{K}_\psi \geq \hat{K}_\phi &\Leftrightarrow \hat{K}_\psi - \hat{K}_\phi \geq 0\\
    &\implies \sum_{j} \Tilde{p}_j \boldsymbol{k}_j\boldsymbol{k}_j^\intercal -
    \sum_{i} p_i \boldsymbol{v}_i\boldsymbol{v}_i^\intercal\geq 0 \label{equ:proof:vertex-GHZ-extremal:condition}
    \intertext{Now we will look at all cases that can happen. The proof is success full when the outcome of the cases is either $\left\vert \phi \right\rangle\sim \left\vert \psi \right\rangle$ or a contradiction.\newline
    \textbf{Case 1}: There exits $\boldsymbol{v}_l\not\in \{\boldsymbol{k}_i\}_i$} 
        0 &\leq \boldsymbol{v}_l^\intercal \left(\sum_{j} \Tilde{p}_j \boldsymbol{k}_j\boldsymbol{k}_j^\intercal -
        \sum_{i} p_i \boldsymbol{v}_i\boldsymbol{v}_i^\intercal\right) \boldsymbol{v}_l\\
        &= - p_l \vert \boldsymbol{v}_l\vert^4 + \sum_{j} \Tilde{p}_j (\boldsymbol{k}_j^\intercal\boldsymbol{v}_l)^2 -
        \sum_{i\neq l} p_i (\boldsymbol{v}_i^\intercal\boldsymbol{v}_l)^2 \\
        &\implies 0 \leq -p_l  + \sum_{j} \Tilde{p}_j \underbrace{\frac{(\boldsymbol{k}_j^\intercal\boldsymbol{v}_l)^2}{\vert \boldsymbol{v}_l\vert^4}}_{<1} -
        \sum_{i\neq l} p_i \underbrace{\frac{(\boldsymbol{v}_i^\intercal\boldsymbol{v}_l)^2}{\vert \boldsymbol{v}_l\vert^4}}_{<1}\\
        &< -p_l + \sum_{j} \Tilde{p}_j  -
        \sum_{i\neq l} p_i = 1 - 1 = 0
        \intertext{\textbf{Case 2}: The strategies $\left\vert\psi\right\rangle$ contains $\{\boldsymbol{k}_i\}_i \supseteq  \{\boldsymbol{v}_i\}_i$ more or the same GHZ states \newline
        We can join the sums where $\boldsymbol{k}_i=\boldsymbol{v}_{i'}$ } 
        (\ref{equ:proof:vertex-GHZ-extremal:condition})&\implies 0\leq \sum_{\{i|\boldsymbol{k}_i = \boldsymbol{v}_{i'}\}} (\Tilde{p}_i-p_{i'}) \boldsymbol{v}_{i'}\boldsymbol{v}_{i'}^\intercal +
        \sum_{\{i|\boldsymbol{k}_i \neq \boldsymbol{v}_{j}\forall j\}} \Tilde{p}_i \boldsymbol{k}_i\boldsymbol{k}_i^\intercal\\
        \intertext{\textbf{Case 2a:} There exits $\Tilde{p}_j \neq p_{j'}$ and $\boldsymbol{k}_j=\boldsymbol{v}_{j'}$, then as they both sumup to 1 there exits $\Tilde{p}_l < p_{l'}$ and $\boldsymbol{k}_l=\boldsymbol{v}_{l'}$}
       &\implies 0 \leq  \sum_{\{i|\boldsymbol{k}_i = \boldsymbol{v}_{i'}\}} (\Tilde{p}_i-p_{i'}) (\boldsymbol{v}_{l'}^\intercal\boldsymbol{v}_{i'})^2 + \sum_{\{i|\boldsymbol{k}_i \neq \boldsymbol{v}_{j}\forall j\}} \Tilde{p}_i (\boldsymbol{k}_i^\intercal\boldsymbol{v}_{l'})^2 \\
       &= (\Tilde{p}_l-p_{l'}) \vert \boldsymbol{v}_{l'}\vert^4 + \sum_{\{i\neq l|\boldsymbol{k}_i = \boldsymbol{v}_{i'}\}} (\Tilde{p}_i-p_{i'}) (\boldsymbol{v}_{l'}^\intercal\boldsymbol{v}_{i'})^2 + \sum_{\{i|\boldsymbol{k}_i \neq \boldsymbol{v}_{j}\forall j\}} \Tilde{p}_i (\boldsymbol{k}_i^\intercal\boldsymbol{v}_{l'})^2 \\
       &\implies 0 \leq (\Tilde{p}_l-p_{l'}) + \sum_{\{i\neq l|\boldsymbol{k}_i = \boldsymbol{v}_{i'}\}} (\Tilde{p}_i-p_{i'}) \underbrace{\frac{(\boldsymbol{v}_{l'}^\intercal\boldsymbol{v}_{i'})^2}{\vert \boldsymbol{v}_{l'}\vert^4}}_{<1} + 
       \sum_{\{i|\boldsymbol{k}_i \neq \boldsymbol{v}_{j}\forall j\}} \Tilde{p}_i \underbrace{\frac{(\boldsymbol{k}_i^\intercal\boldsymbol{v}_{l'})^2}{\vert \boldsymbol{v}_{l'}\vert^4}}_{<1}\\
       &< (\Tilde{p}_l-p_{l'}) + \sum_{\{i\neq l|\boldsymbol{k}_i = \boldsymbol{v}_{i'}\}} (\Tilde{p}_i-p_{i'}) + 
       \sum_{\{i|\boldsymbol{k}_i \neq \boldsymbol{v}_{j}\forall j\}} \Tilde{p}_i =  1 - 1 = 0
        \intertext{Which is a contradiction to being positive.\newline
        \textbf{Case 2b:} $\forall i: \Tilde{p}_i = p_i'$}
        \implies \left\vert \phi \right\rangle \sim \left\vert \psi \right\rangle
\end{align}
\end{proof}

\subsection{Trace criteria for extremal strategies}\label{appendix:max_trace}
In certain cases, e.g. when doing numerical optimization, a simpler condition would be nice to check if a strategy is extremal.
\begin{lemma}\label{lemma:maximal-trace-extremal}
If $\mathrm{tr}(\hat{K})$ of a protocol is maximal, then the protocol is extremal with respect to the partial ordering.
$$\mathrm{tr}(\hat K') = \mathrm{tr}(\hat K) \text{ and } \hat K'\geq \hat K \implies \hat K' = \hat K$$
\end{lemma}
This is particularly useful if the maximal trace is known. For example in the noiseless case the norm of all vertices is $\sqrt{n}$ and the maximal trace is the number of sensors $n$. 
\begin{proof}
We will proof this lemma by contradiction and therefore assume that there exits a protocol $K'\neq K$ such that
\begin{align}
    \hat K' \geq \hat K &\Leftrightarrow \bm v^\intercal \hat K' - \hat K \bm v \geq 0 \forall \bm v.
    \intertext{Now let us look at the trace}
    \mathrm{tr}(\hat{K}') - \mathrm{tr}(\hat{K}) &= \mathrm{tr}(\hat{K}' - \hat{K}) \geq 0
    \intertext{Notice that $ \mathrm{tr}(\hat{K}' - \hat{K})> 0$ is a direct contradiction to the assumption that $\mathrm{tr}(\hat{K})$ is maximal, therefore}
    \mathrm{tr}(\hat{K}') &=  \mathrm{tr}(\hat{K}) \\
    \intertext{Notice that the trace is the same in all othogonal basis and therefore}
    \sum_{v\in ONB} \bm v^\intercal \hat{K}' - \hat{K} \bm = 0
    \intertext{for all ONB. Notice that this is a sum over positive values and hence all summands are zero.}
    &\implies \bm v ^\intercal  \hat{K}' - \hat{K} \bm v = 0 \forall \bm v \\
    &\implies \hat{K} = \hat{K'}
    \intertext{Which is a contradiction to the assumption $K'\neq K$}
\end{align}
\end{proof}


\section{Optimal rates for sequential strategies with orthogonal labels}\label{appendix:orthogonal_labels}
In this section compute the optimal rates for sequential protocols with orthogonal label vectors by using Lagrange multipliers analog to~\cite{bringewattProtocolsEstimatingMultiple2021}.
Therefore we will assume that $\hat{S}$ is invertible. In practice this might not be the case and a pseudo inverse of $\hat S$ projected onto the DFS has to be used. 
If these assumptions are fulfilled the QFIM

\begin{align}
    \cF_\rho^{-1} &= \frac{1}{4 T^2}\hat S^{-1} \sum_i \frac{\bm v_i \bm v_i^\intercal}{r_i |\bm v_i|^4} (\hat S^\intercal)^{-1},
    \intertext{is analytically invertible. The figure of Merit}
    M&= \tr W \cF_\rho^{-1}
    \intertext{simplifies with $W'={S^\intercal}^{-1}WS^{-1}$ to}
    &= \sum_i \frac{1}{|\bm v_i|^4 4 T^2 r_i}\bm v_i^\intercal W' \bm v_i \\
    \intertext{with $w'_i = \bm v_i^\intercal W' \bm v_i/|\bm v_i|^4$}
    &= \frac{1}{4T^2} \sum_i \frac{w_i}{r_i}
    \intertext{Leaving the optimization problem}
    &\min_{r_i} \sum_i \frac{w'_i}{r_i}
    \text{ with }
    \sum_i r_i = 1
\intertext{
To minimize $\sum_{r_i} \frac{w'_i}{r_i}$ under the constraint that $r_i$ is a probability distribution we use the Lagrange multiplies  and minimize $\sum_{i} \frac{w'_i}{r_i} + \lambda (\sum_i r_i -1)$. Deriving with respect to $r_i$ gives}
   &-\frac{w'_i}{r_i^2} + \lambda = 0,
\intertext{or $r_i = \sqrt{\frac{\lambda}{w'_i}}$. By imposing normalization we find that}
r_i &= \frac{1}{\sqrt{w'_i}}\left(\sum_j \frac{1}{\sqrt{w'_j}}\right)^{-1}
\intertext{and}
    M_{opt} = \frac{1}{4T^2} \min_{r_i} \sum_{i}  \frac{w'_i}{r_i} &=  \frac{1}{4T^2} (\sum_i {w'_i}^{3/2}) (\sum_i {w'_i}^{-1/2}).
\end{align}


\section{Second Example for the DFS}
\begin{figure}
    \centering
    \includegraphics[width=0.5\linewidth]{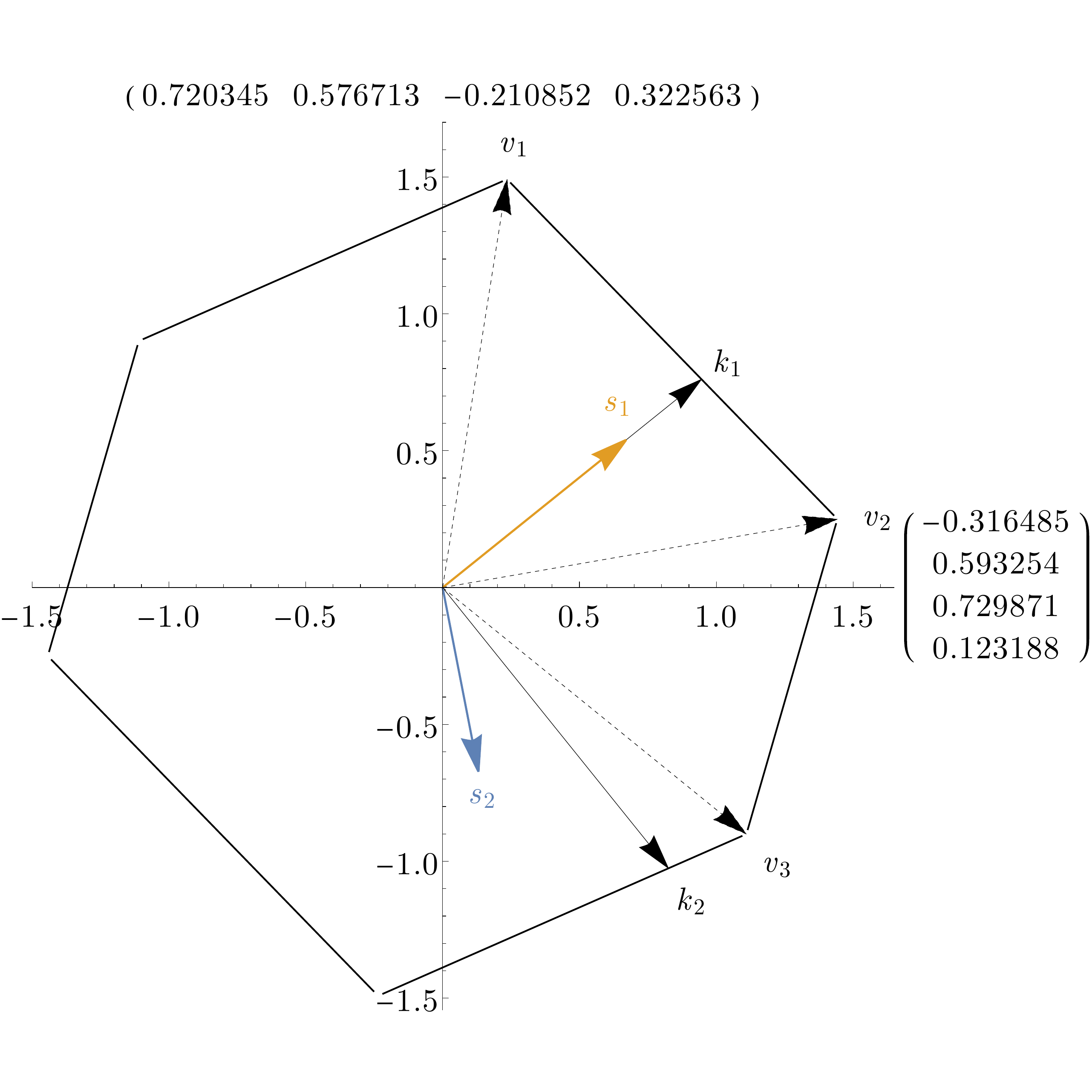}
    \caption{Decoherence free subspace of the non integer example.}
    \label{fig:example-non-integer}
\end{figure}
We will consider an example of a four sensor network with signal matrix $S=\left(\begin{smallmatrix}0.47& 0.47& 0.70& 0.24\\-0.97& 0.24& -0.34& -0.57\end{smallmatrix}\right)$ and noise matrix $N=\left(\begin{smallmatrix}-0.09& -0.36& 0.10& 0.91\\0.80& -0.73& 0.92& 0.12\end{smallmatrix}\right)$. 
The DFS is shown in fig~\ref{fig:example-non-integer} and is given by the convex combinations of 
$\boldsymbol{v}_1=\left(1,1,-0.14,0.51\right)^\intercal$,
$-\boldsymbol{v}_1$,
$\boldsymbol{v}_2=\left(-0.28,1,1,0.26\right)^\intercal$, 
$-\boldsymbol{v}_2$,
$\boldsymbol{v}_3=\left(-1,0.14,1,-0.15\right)^\intercal$ 
and $-\boldsymbol{v}_3$. 

\subsection{Extremal strategy}
The extremal strategy is computed for a given state. Therefore we will consider a simple strategy with a sufficient complexity to show all optimization steps. We will choose $\boldsymbol{k}_1=\boldsymbol{s}_1^{\mathrm{DFS}}=\left(0.25,1,0.53,0.36\right)^\intercal$ to be signal component within the DFS normalized by the infinity norm and $\boldsymbol{k}_2=\left(-1,-0.10,0.82,0.23\right)^\intercal\perp \boldsymbol{k}_1 $ to be normalized an perpendicular.
The strategy is
\begin{align*}
    \left\vert\psi\right\rangle 
    &= \sqrt{p}\left(\sqrt{\frac{1+\Delta}{2}} \vert \boldsymbol{k}_1 \rangle + \sqrt{\frac{1-\Delta}{2}} \left\vert -\boldsymbol{k}_1\right\rangle \right)\\&+\sqrt{1-p}\left(\sqrt{\frac{1+\Delta}{2}} \vert \boldsymbol{k}_2 \rangle + \sqrt{\frac{1-\Delta}{2}} \left\vert -\boldsymbol{k}_2\right\rangle \right),
\end{align*}
with $0\leq p,\Delta\leq 1$. This strategy can be implemented with an auxiliary qubit $\mathcal{A}$. The initialized state is
\begin{align*}&\sqrt{p}\left(\sqrt{\frac{1+\Delta}{2}} \vert0\rangle_{\mathcal{A}}\vert 0,0,0,0\rangle + \sqrt{\frac{1-\Delta}{2}} \left\vert 0\rangle_{\mathcal{A}}\vert 1,1,1,1\rangle\right\rangle \right)\\&+\sqrt{1-p}\left(\sqrt{\frac{1+\Delta}{2}} \vert1\rangle_{\mathcal{A}}\vert 0,0,0,0\rangle + \sqrt{\frac{1-\Delta}{2}} \left \vert1\rangle_{\mathcal{A}}\vert 1,1,1,1\right\rangle \right)\end{align*} where here 0 and 1 labels the computational basis.
Bit flips are then applied at the times $\left(0.64,1,0.76,0.68\right)T$ controlled on the auxiliary system being in $\vert0\rangle$ and $\left(0,0.45,0.91,0.39\right)T$ controlled on the auxiliary system being in $\vert1\rangle$.
Notice that $\vert\psi\rangle$ is within the $\mathcal{DFS}$ and hence protected from the noise.
Following (\ref{equ:K}) its $\hat{K}$ matrix is

\begin{align*}
    \hat{K} &= 
\left(\scalebox{0.65}{$
\begin{array}{cccc}
 \Delta ^2 (p (2.49\, -1.55 p)-1.)-0.93 p+1. & p \left(\Delta ^2 (1.22\, -1.37 p)+0.14\right)-0.10 \Delta ^2+0.10 & 0.81 \Delta ^2+p \left(\Delta ^2 (0.36 p-1.31)+0.95\right)-0.81 & \Delta ^2 (p (0.87\, -0.73 p)-0.22)-0.13 p+0.22 \\
 p \left(\Delta ^2 (1.22\, -1.37 p)+0.14\right)-0.10 \Delta ^2+0.10 & p \left(\Delta ^2 (0.22\, -1.21 p)+0.98\right)-0.01 \Delta ^2+0.01 & 0.08 \Delta ^2+p \left(\Delta ^2 (0.31 p-0.93)+0.61\right)-0.08 & p \left(\Delta ^2 (0.31\, -0.65 p)+0.33\right)-0.02 \Delta ^2+0.02 \\
 0.81 \Delta ^2+p \left(\Delta ^2 (0.36 p-1.31)+0.95\right)-0.81 & 0.08 \Delta ^2+p \left(\Delta ^2 (0.31 p-0.93)+0.61\right)-0.08 & p \left(\Delta ^2 (0.47\, -0.08 p)-0.38\right)-0.66 \Delta ^2+0.66 & \Delta ^2 ((0.17 p-0.55) p+0.18)+0.37 p-0.18 \\
 \Delta ^2 (p (0.87\, -0.73 p)-0.22)-0.13 p+0.22 & p \left(\Delta ^2 (0.31\, -0.65 p)+0.33\right)-0.02 \Delta ^2+0.02 & \Delta ^2 ((0.17 p-0.55) p+0.18)+0.37 p-0.18 & p \left(\Delta ^2 (0.27\, -0.34 p)+0.07\right)-0.05 \Delta ^2+0.05 \\
\end{array}$}
\right)
\intertext{and the QFIM }
    \mathcal{F} &= 4 T^2 \left(
\begin{array}{cc}
 1.09 p-1.09 \Delta ^2 p^2 & 1.23 \Delta ^2 p^2+\left(-0.83 \Delta ^2-0.40\right) p \\
 1.23 \Delta ^2 p^2+\left(-0.83 \Delta ^2-0.40\right) p & -0.63 \Delta ^2-1.40 \Delta ^2 p^2+\left(1.89 \Delta ^2-0.48\right) p+0.63 \\
\end{array}
\right).
\end{align*}
Following Lemma~\ref{lemma:superpositon-GHZ} we get an improved state with $\Delta=0$.  
Due to Lemma~\ref{lemma:sequential-GHZ}, the sequential strategy uses the GHZ states $\vert GHZ_{\boldsymbol{k}_1}\rangle$ and $\vert GHZ_{\boldsymbol{k}_2}\rangle$ with probabilities $p$ and $(1-p)$. 
These vectors can be decomposed into the vertices 
$$\boldsymbol{k}_1 = 0.41 \boldsymbol{v}_1 + 0.59 \boldsymbol{v}_2$$
and 
$$\boldsymbol{k}_1 = 0.21 \left(-\boldsymbol{v}_1\right) + 0.79 \boldsymbol{v}_3$$
Therefore we can use Lemma~\ref{lemma:sequential-vertex-GHZ} to get an improvement by using the sequential vertex protocol using $\vert GHZ_{\boldsymbol{v}_1}\rangle$, $\vert GHZ_{\boldsymbol{v}_2}\rangle$ and $\vert GHZ_{\boldsymbol{v}_3}\rangle$ with probabilities $0.21+0.2p,0.59p$ and $0.79-0.79p$. 
It performance is 
\begin{align*}
    \hat{K} &= \left(
\begin{array}{cccc}
 1.\, -0.54 p & 0.14 p+0.10 & 0.59 p-0.81 & 0.22\, -0.06 p \\
 0.14 p+0.10 & 0.77 p+0.22 & 0.45 p+0.07 & 0.27 p+0.09 \\
 0.59 p-0.81 & 0.45 p+0.07 & 0.79\, -0.19 p & 0.25 p-0.13 \\
 0.22\, -0.06 p & 0.27 p+0.09 & 0.25 p-0.13 & 0.07 p+0.07 \\
\end{array}
\right)
\intertext{ and  }
    \mathcal{F} &= 4 T^2\left(
\begin{array}{cc}
 0.84 p+0.24 & -0.32 p-0.04 \\
 -0.32 p-0.04 & 0.64\, -0.25 p \\
\end{array}
\right).
\end{align*}
as QFIM. Due to lemma~\ref{lemma:vertex-GHZ-extremal} the sequential vertex strategy is extremal and cannot be further improved.
This strategy can be implemented without auxiliary system. Each run a GHZ state is prepared and with probabilities $0.21+0.2p,0.59p$ and $0.79-0.79p$ the local bit flips for $\boldsymbol{v}_1$, $\boldsymbol{v}_2$ or $\boldsymbol{v}_3$ are applied. Summarizing the strategy performs better and is easier to implement as not controlled bit-flip nor distributed auxiliary systems are required.

\subsection{Optimal strategy}
To compute the optimal strategy in this example choose the weights $1$ and $1$ for the signals.
This gives
\begin{equation}
    \min_{p_1+p_2+p_3=1} \mathrm{tr}\left(\begin{pmatrix}
        1 & 0 \\
        0 & 1
    \end{pmatrix}\left(
    \hat{S}
    \sum_{i=1}^3 p_i \boldsymbol{v}_i\boldsymbol{v}_i^\intercal\hat{S}^\intercal
    \right)^{-1}\right)
\end{equation}
as optimization problem.
The optimal solution is then $M_{opt}=3.88$. This is achieved by the probabilities $(0.26,0.30,0.44)$. Optimizing the extreme strategy over $p$ yields $M_{extrem}=0.390$ for $p=0.46$. Optimizing the initial strategy gives $M_{init}=5.26$ with $\Delta=0$ and $p=0.42$.


\section{Affine DFS}\label{appendix:affineDFS}
\subsection{Example affine DFS is advantageous}
\begin{figure}
    \centering
    \includegraphics[width=0.4\columnwidth]{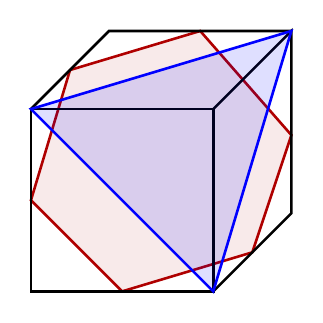}
    \caption{The standard DFS (red) and the affine one DFS' (blue). The vertices of the affine DFS are vertices of the original orthotope.}
    \label{fig:affine_DFS}
\end{figure}

Let us consider a three sensor network. Let there be one noise source $(1,1,1)$, which is constant. 
The DFS is then given by the polytope
\begin{align}
    \bm v_1 &= \begin{pmatrix}
    1 \\ 0 \\ -1
    \end{pmatrix} = - \bm v_4 &
    \bm v_2 &= \begin{pmatrix}
    -1 \\ 1 \\ 0
    \end{pmatrix} = - \bm v_5 &
      \bm v_3 &= \begin{pmatrix}
    0\\ -1 \\ 1
    \end{pmatrix} = - \bm v_6
\intertext{Additionally we will look at the affine $\mathrm{DFS'}=\mathrm{DFS}_{(1,0,0)}$, with the vertices}
    \bm v'_1 &= \begin{pmatrix}
    1 \\ 1 \\ -1
    \end{pmatrix}&
    \bm v'_2 &= \begin{pmatrix}
    -1 \\ 1 \\ 1
    \end{pmatrix}&
      \bm v'_3 &= \begin{pmatrix}
    1\\ -1 \\ 1
    \end{pmatrix}
\end{align}
These spaces are shown in Fig.~\ref{fig:affine_DFS}.

Now we will use the equal superposition of all vertices as strategy. This state is a W State~\cite{PhysRevA.62.062314} $1/{\sqrt 3}(|001\rangle + |010\rangle + |100\rangle)$ in standard notation, which obtains a non-zero global phase from a constant signal.

Thereby we get
\begin{align}
    K &= \left(
\begin{array}{ccc}
 \frac{2}{3} & -\frac{1}{3} & -\frac{1}{3} \\
 -\frac{1}{3} & \frac{2}{3} & -\frac{1}{3} \\
 -\frac{1}{3} & -\frac{1}{3} & \frac{2}{3} \\
\end{array}
\right)&
    K' &= \left(
\begin{array}{ccc}
 \frac{8}{9} & -\frac{4}{9} & -\frac{4}{9} \\
 -\frac{4}{9} & \frac{8}{9} & -\frac{4}{9} \\
 -\frac{4}{9} & -\frac{4}{9} & \frac{8}{9} \\
\end{array}
\right) = \frac{4}{3} K > K
\end{align}
as matrices. Even though this is just one example. Notice that no strategy in the standard, non-affine DFS can obtain $\tr K > 2$, therefore the affine DFS is more powerful.

\subsubsection{Single parameter estimation}
For the single parameter discussion we will use the signal $(1,-1,0)$.
The optimal state within the DFS is $\vert \mathrm{GHZ}_{\bm v_3}\rangle$ and obtains a QFI of $4 (\bm s \bm v_3)^2 = 16$. 
The equal superposition of vertices in $\mathrm{DFS'} = 4 \bm s ^\intercal K' \bm s = 10.67$. The optimal superposition of vertices is $\frac{\vert \bm v'_2\rangle + \vert \bm v'_3\rangle  }{\sqrt{2}}$ and reaches the same QFI as $\vert \mathrm{GHZ}_{\bm v_3}\rangle$. As expected~\cite{sekatskiOptimalDistributedSensing2020} the affine DFS is not advantageous in the single parameter case.

\subsection{Upper-bounding the advantage}\label{appendix:affineDFS:upperbound}
Given any strategy $\vert \psi \rangle = \sum_{\bm k \in \mathrm{DFS}_{\bm \kappa}} p_{\bm k} \vert \bm k \rangle$ in the affine DFS, we define the strategy $\vert \phi \rangle = \sum_{\bm k \in \mathrm{DFS}_{\bm \kappa}} p_{\bm k} \vert \frac{\bm k-\bm \bar{\bm k}}{2} \rangle$, where $\Bar{k}=\sum_{\bm k \in \mathrm{DFS}_{\bm \kappa}} p_{\bm k}\bm k$.
First we will show that $\vert\phi\rangle$ is a valid strategy within the DFS then we will show that $\hat{K}_\phi = \frac{\hat{K}_\psi}{4}$ and hence going to the affine DFS gives at most an improvement by a factor of 4.

To show that $\vert \phi \rangle$ is a valid strategy we will investigate properties of $\mathrm{DFS}{\bm \kappa}$ and its "opposite"  $\mathrm{DFS}{-\bm \kappa}$.
First notice that every element in the $\mathrm{DFS}_{\bm \kappa}$ can be decomposed $\bm k = \bm k_\parallel + \bm \kappa$ as a sum of a vector $\bm k_\parallel $ in the $\mathrm{Kernel}(\hat{N})$ and some orthogonal vector $\bm \kappa_\perp$ with $\hat{N}\bm \kappa=\bm \kappa$. In addition $\|\bm k \|_\infty = \|-\bm k \|_\infty$, and hence for every element $\bm k\in \mathrm{DFS}_{\bm \kappa}$ the inverse is in the opposite DFS $-\bm{k} \in \mathrm{DFS}_{-\bm \kappa}$.  
Additionally, for any two vector $\bm a\in \mathcal{DFS}_{\bm \kappa}$ and $\bm b \in \mathrm{DFS}_{-\bm \kappa}$ the average $\frac{\bm a+\bm b}{2}\in \mathrm{DFS} \forall$ is within the original DFS, i.e. $\hat{N}\frac{\bm a+\bm b}{2}=0$ and $ \left \|\frac{\bm a+\bm b}{2} \right\|_\infty \leq 1$. The norm condition is fulfilled because it is a convex combination. The kernel condition is fulfilled as $\bm a+\bm b=\bm a_\parallel +\bm \kappa +\bm b_\parallel -\bm \kappa = \bm a_\parallel+\bm b_\parallel$ the component orthogonal to the kernel is cancelled. Applying the last property to our strategy $\vert\phi\rangle,$ with $\bm a=\bm k$ and $\bm b=-\Bar{\bm k}$ shows that it is contained within the DFS.

Finally the the performance of $\vert\phi\rangle$ is given by

\begin{align}
    \hat{K}_\phi &= \sum_{\bm k} p_{\bm k} \left(\frac{\bm k - \Bar{\bm k}}{2}\right)\left(\frac{\bm k - \Bar{\bm k}}{2}\right)^\intercal \\
    &= \frac{1}{4}\left(\sum_{\bm k} p_{\bm k}\bm k\bm k^\intercal - \sum_{\bm k} p_{\bm k}\bm k\bar{\bm k}^\intercal -\bar{\bm k}\bar{\bm k}^\intercal+\bar{\bm k}\bar{\bm k}^\intercal\right)\\
     &= \frac{\hat{K}_\psi }{4}
\end{align}
As this holds for any strategy in $\mathrm{DFS}_\kappa$ the improvement by going into the affine DFS is at most 4. 

\end{document}